\documentclass[prx,aps,twocolumn, amssymb,amsmath,floatfix,superscriptaddress]{revtex4-2}

\usepackage{graphicx}
\usepackage{dcolumn}
\usepackage{bm}
\usepackage{hyperref}
\usepackage{amsthm}

\usepackage{physics} 
\usepackage{comment}
\usepackage{lineno}
\usepackage{braket}
\usepackage{tabularx}
\usepackage{bm}
\usepackage{euscript}
\usepackage{epsfig,psfrag}
\usepackage{amsthm}
\theoremstyle{definition}

\usepackage{graphicx}
\usepackage{color}
\usepackage{amsfonts}
\usepackage{exscale}
\usepackage{wrapfig}
\usepackage{amsthm}
\newtheorem{lemma}{Lemma}
\usepackage{extarrows} 
\usepackage{placeins}
\usepackage{soul}
\usepackage{amsmath}
\usepackage[english]{babel}
\usepackage[autostyle, english = american]{csquotes}
\MakeOuterQuote{"}
\usepackage{enumitem}
\usepackage{slashed} 
\usepackage{tikz}
\numberwithin{equation}{section} 
\renewcommand{\theequation}{\arabic{section}.\arabic{equation}} 
\newcommand{\bea}{\begin{eqnarray}}  
\newcommand{\eea}{\end{eqnarray}}
\newcommand{\ben}{\begin{enumerate}}
\newcommand{\een}{\end{enumerate}}
\newcommand{\be}{\begin{equation}}
\newcommand{\ee}{\end{equation}}

\newtheorem{theorem}{Theorem}
\newtheorem{corollary}[theorem]{Corollary}

\definecolor{Bcolor}{RGB}{0,0,255}
\definecolor{Rcolor}{RGB}{255,0,0}
\definecolor{Gcolor}{RGB}{0,255,0}

\usetikzlibrary{external}
\begin{document}

\title{R\'enyi Phase Transitions and Analytic Continuation to the von Neumann Entropy}

\author{Ayush Raj}
\affiliation{Department of Physics and Astronomy, Purdue University, West Lafayette, IN 47907, USA}

\author{Akash Vijay}
\affiliation{Department of Physics and Institute of Condensed Matter Theory, University of Illinois at Urbana-Champaign, Urbana, Illinois 61801, USA}

\author{Hong-Chen Jiang}
\affiliation{Stanford Institute for Materials and Energy Sciences, SLAC National Accelerator Laboratory, Menlo Park, California 94025, USA}

\author{Laimei Nie}
\affiliation{Department of Physics and Astronomy, Purdue University, West Lafayette, IN 47907, USA}

\date{\today}

\begin{abstract}
Extracting operationally meaningful quantities such as von Neumann entropy and mutual information is central to characterizing many-body quantum systems, yet experiments and numerics often provide direct access only to integer R\'enyi entropies. Whereas conventional extrapolation relies on a prescribed fitting ansatz, in Ref.~\cite{vijay2026analytically} we introduced an alternative approach based on stabilized analytic continuation (SAC), which avoids such an ansatz and is inherently robust to noise. Here, we further develop this framework in the noiseless setting and benchmark its performance across several nontrivial many-body problems. 
We first show that, besides the intrinsic ambiguity of reconstructing the von Neumann entropy from finitely many R\'enyi samples, analytic continuation can also fail because of genuine nonanalyticities in the R\'enyi function arising from zeros of $\Tr\rho^z$. We derive universal zero-free domains for $\Tr\rho^z$ and systematically sharpen these bounds using additional spectral information about $\rho$.
We then benchmark the performance of SAC in three settings: ($i$) extracting topological entanglement entropy from DMRG Rényi data for the toric code and Kagome Heisenberg models; ($ii$) reconstructing the mutual information between disjoint intervals in a $(1+1)d$ compact-boson CFT; and ($iii$) recovering finite-time ballistic growth of the von Neumann entropy from integer R\'enyi entropies that cross over toward subballistic growth in diffusive quantum dynamics. In the latter setting, analytic continuation is expected to fail in the long time limit. We illustrate this mechanism with a physically motivated two-sector model, where the asymptotic separation between von Neumann and higher-R\'enyi growth is accompanied by a zero of $\Tr\rho^z$ pinching the real axis at $z=1$, thereby producing a first-order R\'enyi phase transition.
\end{abstract}
                              
\maketitle

\makeatletter
\newcommand{\hideappendixsubsections}{%
  \renewcommand{\l@subsection}[2]{}%
  \renewcommand{\l@subsubsection}[2]{}%
}
\makeatother

\makeatletter
\begingroup
\renewcommand{\l@subsubsection}[2]{}
\tableofcontents
\endgroup
\makeatother

\begin{figure*}[!t]
\centering
\includegraphics[width=\textwidth]{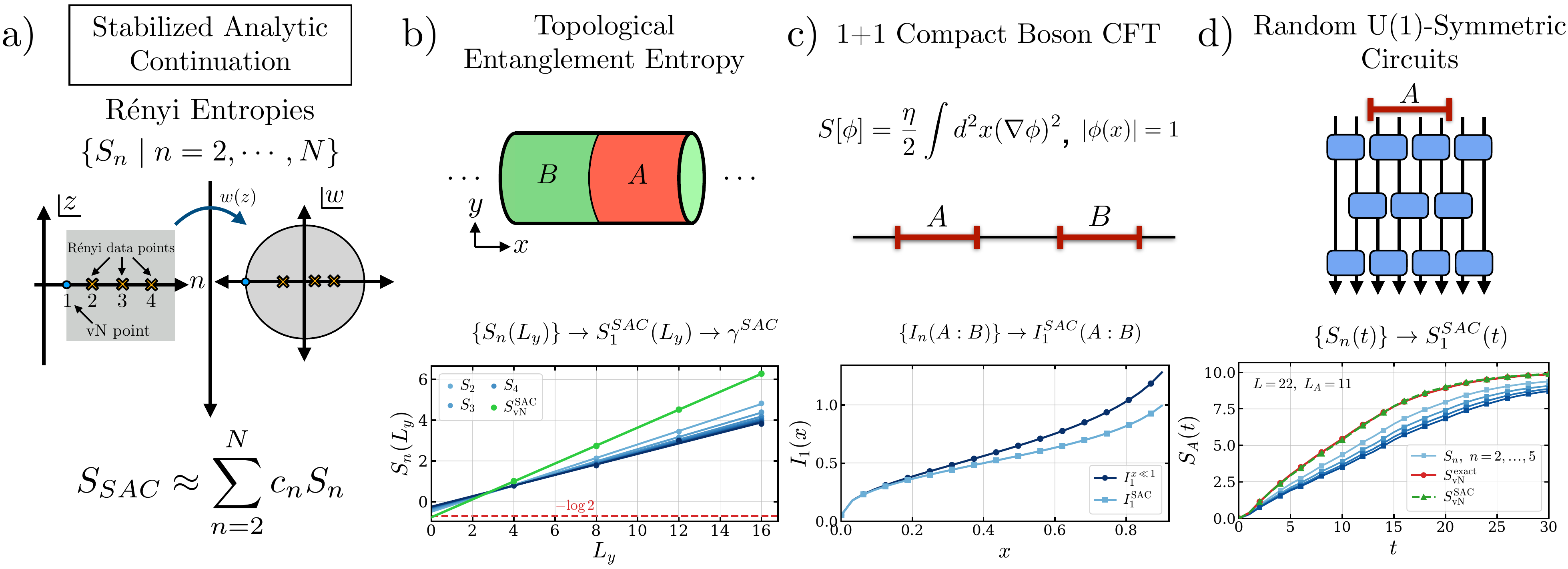}
\caption{\textbf{Stabilized analytic continuation of R\'enyi entropies and representative applications.}
(a) Schematic of the noiseless SAC construction. Integer R\'enyi data are mapped from the complex R\'enyi-index plane to the unit disc, where SAC selects a minimum-norm analytic continuation to the von Neumann point $n=1$. For a fixed set of R\'enyi entropies and conformal-map parameters, the resulting estimator is a state-independent linear combination,
$S_{\mathrm{vN}}^{\mathrm{SAC}}=\sum_n c_n S_n$.
(b) Extraction of the topological entanglement entropy in the spin-$1/2$ Kagome Heisenberg model on a cylinder with $L_x=48$, $J_1=1$, and $J_2=0.1$. At each circumference $L_y$, SAC reconstructs $S_{\mathrm{vN}}(L_y)$ from $S_2,\ldots,S_{10}$. Fitting the result to $S_{\mathrm{vN}}^{\mathrm{SAC}}(L_y)=\alpha L_y-\gamma$ yields $\gamma^{\mathrm{SAC}}\simeq0.756$. The red dashed line marks the expected $\mathbb{Z}_2$ intercept $-\log 2$.
(c) Reconstruction of the mutual information between two disjoint intervals in the $(1+1)d$ compact-boson CFT. SAC continues the integer-R\'enyi mutual informations $I_n(A:B)$, $n=2,\ldots,7$, to $n=1$ for compactification parameter $\eta=3$, and is compared with the analytic small-cross-ratio result.
(d) Entanglement growth in a random $U(1)$-symmetric brick-wall circuit with $L=22$, total charge $Q=11$, and subsystem size $L_{A}=L/2=11$. The SAC reconstruction from $S_2,\ldots,S_5$, averaged over $50$ circuit realizations, closely tracks the directly computed von Neumann entropy throughout the simulated time window.}
\label{fig:panel_figure}
\end{figure*}

\section{Introduction}  

The quantum-information perspective has become central to the characterization of many-body quantum systems, providing diagnostics of phenomena ranging from topological order and quantum criticality to thermalization, scrambling, and many-body localization~\cite{TEE_KitaevPreskill,TEE_LevinWen,Fast_Scramblers,ETH_EE,Deutsch_2018,Kaufman_2016,EntanglementGrowthRUCs,EE_Criticality,EE_PhaseTransition,Ryu_2006,Ryu_2007,abanin2019colloquium,nandkishore2015many}. Among the most fundamental quantities is the von Neumann entropy,
\begin{equation}
    S_{\mathrm{vN}}(\rho)=-\mathrm{Tr}(\rho\log\rho),
\end{equation} 
together with its cousins such as mutual information. Despite their conceptual and operational significance, these quantities are difficult to access directly in both experiments and numerical many-body calculations~\cite{Abanin2012PRL,haah2017sample,ODonnell2016Jun}.

By contrast, integer R\'enyi entropies,
\begin{equation}
    S_n(\rho)=\frac{1}{1-n}\log\mathrm{Tr}(\rho^n),
    \qquad n=2,3,\ldots,
\end{equation}
are naturally accessible through measurements of integer moments of the density matrix. Randomized-measurement and shadow-tomography protocols can determine low-order moments experimentally~\cite{Daley2012PRL,islam2015measuring,Kaufman_2016,elben2019statistical,huang2020predicting,elben2023randomized}, while replica-based constructions provide direct access to integer R\'enyi entropies in quantum Monte Carlo and related numerical approaches~\cite{Hastings2010PRL,Humeniuk2012QMC,Chung2014PRBQMC, Luitz2014PRBQMC, Assaad2014PRBQMC, Meng2022QMC}. This leaves a fundamental gap between what is readily accessible and what is ultimately desired. In practice, one typically has access to only a few integer R\'enyi entropies at $n\geq2$, and the von Neumann entropy
is obtained from the formal limit $S_{\mathrm{vN}}(\rho)=\lim_{n\to 1} S_n(\rho)$.
Therefore, inferring this limit from finitely many integer values is  a problem of analytic continuation. Two distinct difficulties arise.

First, finitely many samples do not uniquely determine an analytic function, even when the data are exact and noiseless. Conventional approaches address this ambiguity by fitting the R\'enyi function to a prescribed polynomial, rational polynmial, or other finite-dimensional ansatz~\cite{zyczkowski2005renyiextrapolationshannonentropy,Headrick2014pade, Hoker2021Taylor,vermersch2024enhanced}. 
With only a small number of input points, however, the resulting continuation can depend sensitively on the chosen functional form. 
In Ref.~\cite{vijay2026analytically}, we introduced an alternative approach based on the stabilized analytic continuation (SAC) framework of Ciulli and Spearman~\cite{CS1,CS2,CS3,CS4,CS5}. 
Rather than assuming a particular form for $S_n$, SAC uses analyticity itself as the organizing principle. 
After conformally mapping the relevant analytic domain to the unit disc, the method selects, among all functions consistent with the available data, the one that minimizes an appropriate boundary norm.
The framework also naturally incorporates uncertainty in the input data, making it well suited to experimentally accessible, noisy R\'enyi measurements~\cite{vijay2026analytically}.

Here, we focus on the complementary noiseless setting, in which SAC reduces to a closed-form estimator for the von Neumann entropy~\cite{vijay2026analytically}. For a fixed set of input R\'enyi entropies $\{S_n\,|\,n=2,\ldots,N+1\}$, the estimate takes the form
\begin{align}\label{eq:SAC_estimator}
    S_{\mathrm{vN}}^{\mathrm{SAC}}
    &=
    \frac{
    \displaystyle
    \sum_{i,j=3}^{N+1}(A^{-1})_{ij}
    \left(\frac{S_i(\rho)}{i-1}-S_2(\rho)\right)
    \left(\frac{1}{j-1}-1\right)}
    {\displaystyle
    \sum_{i,j=3}^{N+1}(A^{-1})_{ij}
    \left(\frac{1}{i-1}-1\right)
    \left(\frac{1}{j-1}-1\right)}
    \nonumber\\
    &=
    \sum_{n=2}^{N+1} c_n S_n .
\end{align}
The coefficients $c_{n}$ are determined entirely by the $(N-1)\times(N-1)$ positive-semidefinite symmetric matrix $A$, whose entries are given by
\begin{align}
    \label{eq:A_matrix_elements}
    A_{ij}
    =
    \frac{2}{\pi}
    \int_0^{2\pi}
    \ln\left|
    \frac{e^{i\theta}-w(i)}
         {e^{i\theta}-w(2)}
    \right|
    \ln\left|
    \frac{e^{i\theta}-w(j)}
         {e^{i\theta}-w(2)}
    \right|
    d\theta ,
\end{align}
where $w(z)$ is a conformal map which will be specified below. Importantly, $A$ depends only on the chosen R\'enyi-index set $\{2,\ldots,N+1\}$ and the conformal map, not on the underlying quantum state. Thus, once these are fixed, the coefficients $\{c_n\}$ can be computed once and applied universally to any input state.

The second difficulty concerns the analytic structure of the R\'enyi function itself. Even if the finite-data ambiguity could be overcome, the R\'enyi function $S_z(\rho)$ itself may develop singularities that can obstruct continuation. These singularities originate from zeros of $\Tr\rho^z$ in the complex $z$ plane~\cite{Gliozzi_2010}, which become branch points of $S_z(\rho)$. At finite rank, $\Tr\rho^z$ is entire and strictly positive for real $z>0$, so $S_z(\rho)$ is analytic in a neighborhood of every point on the positive real axis. As the rank of $\rho$ grows, however, zeros of $\Tr\rho^z$ can accumulate and pinch the positive real axis, producing a \emph{R\'enyi phase transition}. If the resulting singularity occurs at $1\leq z<2$, it poses a genuine obstruction to analytic continuation from the integer-R\'enyi entropies to the von Neumann entropy.

We begin by determining how closely these zeros can approach the continuation region. For every density matrix of rank $r\geq3$, we show that $\Tr\rho^z$ is zero-free in the universal wedge
\begin{equation}
    \mathcal W_r
    =\left\{
        z=u+iv:\ u>0,\quad
        |v|<\frac{\pi u}{\log(r-1)}
      \right\}.
    \label{eq:intro_zero_free_wedge}
\end{equation}
Within the continuation half-plane $\Re z\geq1$, this implies that no zero can lie closer to the real axis than $\pi/\log(r-1)$. This bound is saturated by a spectrum with one distinguished eigenvalue carrying half of the total weight and a degenerate sector carrying the other half. 
In the large-rank limit, competition between these sectors produces a first-order R\'enyi phase transition, providing a concrete spectral mechanism for obstructing analytic continuation.
We then sharpen this universal result using additional spectral information, including the modular spectral width, the entropy gap $S_{\mathrm{vN}}-S_\infty$, and fidelities to pure and maximally mixed states. These bounds provide increasingly refined criteria for excluding nearby singularities that can obstruct
continuation.

Finally, we benchmark SAC in three complementary many-body settings, summarized in Fig.~\ref{fig:panel_figure}. 
First, we reconstruct the von Neumann entropy from DMRG R\'enyi data for the toric-code and Kagome Heisenberg models and use it to extract the topological entanglement entropy, with direct DMRG results providing a benchmark. Second, we apply SAC to the R\'enyi mutual information of two disjoint intervals in a $(1+1)d$ compact-boson CFT, where the integer-R\'enyis are known exactly but the $n\to1$ continuation is highly nontrivial. Third, we explore the applicability of SAC to entanglement growth in chaotic systems with conservation laws, where diffusion produces subballistic higher-R\'enyi growth
while the von Neumann entropy remains ballistic
~\cite{rakovszky2019sub,Huang2020qudit}. We further introduce a physically motivated, time-dependent two-sector model~\cite{rakovszky2019sub} that reproduces the asymptotic growth laws $S_{\mathrm{vN}}\sim v_E t$ and $S_{n>1}\propto\sqrt{t}$. In this model, the zeros pinch the real axis at $z=1$ in the long-time limit, producing a first-order R\'enyi transition. This provides a spectral explanation for how SAC can remain accurate at accessible times yet ultimately fail asymptotically.

The paper is organized as follows. Sec.~\ref{sec:domain_analyticity} develops the interpretation of the zeros, proves the zero-free bounds, and studies their consequences for many-body spectra. Sec.~\ref{sec:SAC} derives the SAC
construction and specifies the conformal map. Sec.~\ref{sec:many_body_benchmarks} presents the
many-body benchmarks. And we conclude in Sec.~\ref{sec:conclusions} with a discussion of broader implications. Proofs and implementation details are collected in the appendices.

\section{Domain of Analyticity of $S_{z}(\rho)$}
\label{sec:domain_analyticity}

The singularities of $S_z(\rho)$ determine the analytic domains available for continuation from integer R\'enyi entropies to the von Neumann limit. 
These singularities have a direct physical interpretation: they arise from
Fisher zeros~\cite{Fisher1965, meixner1965statistical} of an auxiliary partition function associated with the modular Hamiltonian. In the thermodynamic limit, their accumulation toward the real R\'enyi axis can produce a R\'enyi phase transition and obstruct analytic continuation. In this section, we first make this correspondence precise, then derive a hierarchy of zero-free bounds, and finally identify spectral structures that either drive Fisher zeros toward the continuation region or
keep them parametrically far away.

For a finite-rank density matrix $\rho$, define
\begin{equation}
    f_\rho(z)\equiv\Tr\rho^z, \qquad S_z(\rho)=\frac{\log f_\rho(z)}{1-z}.
\end{equation}
Here and throughout, traces and operator functions are restricted to the support of $\rho$. On this support, $f_\rho(z)$ is an entire function of $z$ and is strictly positive along the positive real axis. The nontrivial singularities of $S_z(\rho)$ instead arise at the complex zeros of $f_\rho(z)$, where the logarithm develops branch points. Consequently, determining where $S_z(\rho)$ is analytic amounts to finding zero-free regions of $f_\rho(z)$.

To understand the physical meaning of these zeros, we introduce the modular Hamiltonian
\begin{equation}
    K_\rho\equiv-\log\rho ,
\end{equation}
and, for real $u>0$, the tilted state
\begin{equation}
    \rho_u
    \equiv
    \frac{\rho^u}{\Tr\rho^u}.
\end{equation}
$u$ plays the role of an inverse temperature for the modular Hamiltonian. We denote expectation values in $\rho_u$ by
$\langle O\rangle_u\equiv\Tr(\rho_u O)$ and define $\delta K_{\rho,u}\equiv K_\rho-\langle K_\rho\rangle_u$. For a complex R\'enyi index $z=u+iv$,
\begin{align}
    f_\rho(u+iv)
    &=
    f_\rho(u)\,
    \Tr\!\left(
        \rho_u e^{-ivK_\rho}
    \right)
    \nonumber\\
    &=
    f_\rho(u)\,
    e^{-iv\langle K_\rho\rangle_u}\,
    \chi_u(v),
    \label{eq:f_rho_characteristic}
\end{align}
where
\begin{equation}
    \chi_u(v)
    \equiv
    \Tr\!\left(
        \rho_u e^{-iv\delta K_{\rho,u}}
    \right)
    \label{eq: zeros of characteristic function}
\end{equation}
is the characteristic function of the modular-energy distribution $\delta K_{\rho, u}$ in $\rho_u$. Since the first two factors in Eq.~\eqref{eq:f_rho_characteristic} never vanish, the zeros of $f_{\rho}(u + iv)$ will be the zeros of $\chi_u(v)$. They arise when contributions from different parts of the modular spectrum acquire phases that cancel
destructively.

\subsection{Fisher Zeros and R\'enyi Phase Transitions}
\label{sec:fisher-physical-interpretation}

Because the moment function $f_{\rho}(z)$ is the partition function of the modular Hamiltonian at complex inverse temperature
\begin{equation}
    \Tr e^{-zK_\rho}
    =
    f_\rho(z),
    \label{eq:entanglement-partition-function}
\end{equation}
its complex zeros are therefore Fisher zeros of the auxiliary modular ensemble~\cite{Fisher1965}.
For a finite system, Fisher zeros remain away from the positive real axis. In the thermodynamic limit, the leading zeros may approach a real value $u_c$ and pinch the real axis, producing a nonanalyticity in the modular free energy. We refer to such a singularity as a \emph{R\'enyi phase
transition}.

For our purpose of analytic continuation, the most important scenario would be
\begin{equation}
    1 \le u_c<2.
    \label{eq:renyi_transition_interval}
\end{equation}
After the thermodynamic limit is taken, the von Neumann and integer R\'enyi  regimes are separated by a nonanalyticity and cannot be connected by a single analytic branch. At finite size the positive real axis remains analytic, but the approach of Fisher zeros toward it prevents a size-independent analytic domain.

We emphasize that a R\'enyi transition is a transition of the auxiliary modular ensemble as $u$ is varied and does not necessarily coincide with a quantum phase transition of a microscopic parent Hamiltonian
~\cite{Chandran2014,Metlitski2009}.

\subsection{Zero-free bounds}
\label{sec:zero-free-domains}

Having identified the zeros of $f_\rho(z)$ as Fisher zeros of the modular ensemble, we now ask how closely they can approach the real R\'enyi axis. For the continuation problem, it is convenient to define the Fisher-zero clearance
\begin{equation}
    \Delta_{\mathrm F}(\rho)
    \equiv
    \inf_{\substack{f_\rho(z)=0\\ \Re z\geq1}}
    |\Im z|,
    \label{eq:fisher_clearance}
\end{equation}
with $\Delta_{\mathrm F}=\infty$ if no Fisher zero occurs in the continuation half-plane $\Re z \ge 1$. We now derive a hierarchy of lower bounds on
$\Delta_{\mathrm F}$, requiring progressively more information about the spectrum.

Let $d = \text{dim}(\mathcal{H})$ be the Hilbert space dimension and let $r$ denote the rank of the density matrix $\rho$, hence $r \le d $. We begin with a state-independent result that depends only on the rank $r$.

\begin{theorem}[Universal zero-free wedge]
\label{thm:universal_zero_free_wedge}
Let $\rho$ have rank $r\geq3$. Then
$f_\rho(z)=\Tr\rho^z$ is zero-free in
\begin{equation}
    \mathcal W_r
    =
    \left\{
        z=u+iv:\;
        u>0,\quad
        |v|<
        \frac{\pi u}{\log(r-1)}
    \right\}.
    \label{eq:universal_zero_free_wedge}
\end{equation}
For $r\leq2$, $f_\rho(z)$ is zero-free throughout $u = \Re z>0$.
The bound is sharp.
\end{theorem}

\begin{figure}
    \centering
    \includegraphics[width=0.99\linewidth]{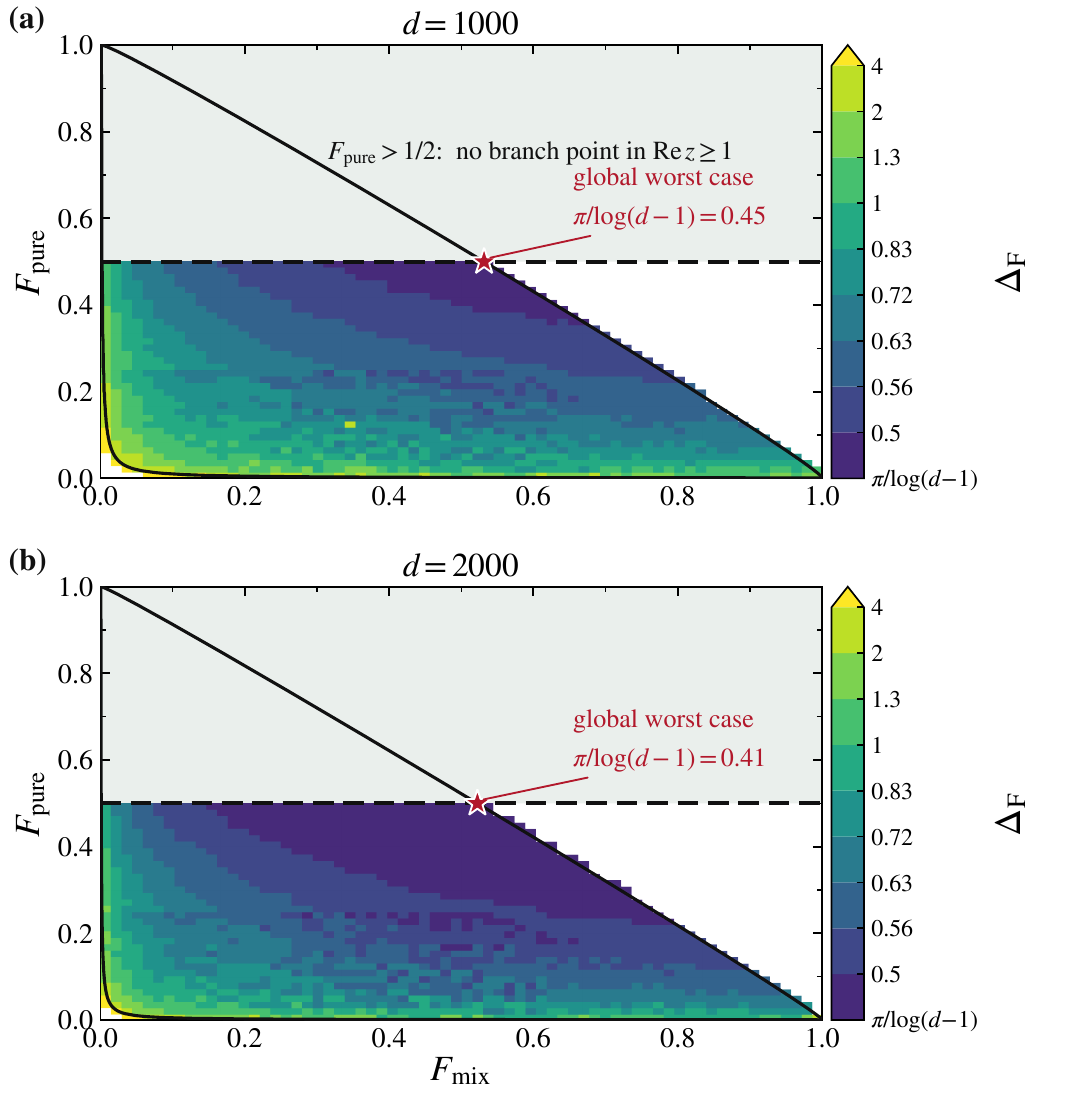}
    \caption{\textbf{Fisher-zero clearance in the joint-fidelity plane.}
    Empirical minimum of the Fisher-zero clearance $\Delta_{\mathrm F}$ over sampled spectra in the $(F_{\mathrm{mix}},F_{\mathrm{pure}})$ plane, for $d=1000$ (top) and $d=2000$ (bottom).
    Solid curves delimit the kinematically accessible fidelity region, whose upper boundary approaches $F_{\mathrm{mix}}+F_{\mathrm{pure}}=1$ as $d\to\infty$.
    For $F_{\mathrm{pure}}>1/2$, Theorem~\ref{thm:pure_fidelity_confinement} excludes all Fisher zeros from $\Re z\geq1$.
    The star marks the balanced two-sector state $\rho_\star$, which attains the exact global minimum $\Delta_{\mathrm F}=\pi/\log(d-1)$. The color scale shows the smallest clearance observed among spectra sampled within each fidelity bin and should therefore be interpreted as an empirical lower envelope rather than a rigorous conditional minimum.}
    \label{fig:fidelity_bounds}
\end{figure}

Because the wedge widens linearly with $u$, its narrowest cross-section in the continuation half-plane $u\geq1$ occurs at $u=1$. Hence
$\mathcal W_r$ contains the constant-width strip
\begin{equation}
    u \geq 1,
    \qquad
    |v|<
    \frac{\pi}{\log(r-1)}.
    \label{eq:universal_SAC_strip}
\end{equation}
Equivalently,
\begin{equation}
    \Delta_{\mathrm F}(\rho)
    \geq
    \frac{\pi}{\log(r-1)}.
\end{equation}
Thus no rank-$r$ state can have a Fisher zero closer to the real axis than $\pi/\log(r-1)$ within this half-plane. In Sec.~\ref{sec:spectral-coexistence}, we identify the density matrix that saturates this worst-case bound; many other spectra can have a substantially larger zero-free region.

The proof of Theorem 1 is given in Appendix~\ref{app:universal_wedge_proof}, but its physical origin is simple. A Fisher zero requires destructive interference among the modular-energy contributions to $\chi_u(v)$ in~\eqref{eq: zeros of characteristic function}: higher modular energies accumulate larger relative phases with $|v|$, but their tilted weights are exponentially suppressed by $u$. With only $r$ levels available, cancellation would be impossible unless $|v|/u$ is sufficiently large, giving the wedge-shaped bound.

Theorem 1 is universal because it only invokes the rank. One can obtain a sharper bound if more information about the spectrum is known. Define the total width of the modular spectrum,
\begin{equation}
    \Delta_K(\rho)
    \equiv
    \lambda_{\max}(K_\rho)-\lambda_{\min}(K_\rho)
    =
    \log\!\bigl(
        \|\rho\|_\infty\|\rho^{-1}\|_\infty
    \bigr),
    \label{eq:modular_width}
\end{equation}
where $\| \cdot \|_{\infty}$ is the $L^{\infty}$ norm, and $\rho^{-1}$ denotes the inverse on the support of $\rho$. This spectral width limits the relative phases that can develop between modular energies, leading to the following bound.

\begin{theorem}[Spectral-width bound]
\label{thm:spectral_width_bound}
For $\Delta_K(\rho)>0$, every Fisher zero $z = u+iv$ satisfies
\begin{equation}
    |v|
    \geq
    \frac{\pi}{\Delta_K(\rho)}.
    \label{eq:spectral_width_bound}
\end{equation}
If $\Delta_K(\rho)=0$, the spectrum is flat and $f_\rho(z)$ has no zeros.
\end{theorem}

Proof of Theorem 2 can be found in Appendix~\ref{app:spectral_width_proof}. Geometrically, as long as $|\Im z|\,\Delta_K<\pi$, the phases accumulated across the entire modular spectrum span less than a semicircle, so their weighted sum cannot vanish. A large spectral width is therefore necessary for a Fisher zero to lie close to the real axis. 

Besides relative phases, cancellation also relies on the weights carried by widely separated parts of the spectrum. This motivates a complementary bound that incorporates spectral weight.

\begin{theorem}[Entropy-gap bound]
\label{thm:entropy_gap_clearance}
Every Fisher zero $z=u+iv$ with $u\geq1$ satisfies
\begin{equation}
    |v|
    \geq
    \frac{\pi}{
        2\left[
            S_{\mathrm{vN}}(\rho)-S_\infty(\rho)
        \right]
    }.
    \label{eq:entropy_gap_clearance}
\end{equation}
If $S_{\mathrm{vN}}=S_\infty$, then the spectrum is flat on its support and $f_\rho(z)$ is zero-free everywhere.
\end{theorem}

Proof of Theorem 3 can be found in Appendix~\ref{app:entropy_gap_proof}. Since $S_\infty(\rho) \equiv  -\log\|\rho\|_\infty = \lambda_{\min} (K_{\rho})$ is the lowest modular energy, the entropy gap
\begin{equation}
    S_{\mathrm{vN}}(\rho)-S_\infty(\rho) = \Tr \big[\rho \big(K_{\rho} - \lambda_{\min}(K_{\rho}) \big) \big]
    \label{eq:entropy_gap}
\end{equation}
is the average modular energy above the ground level measured in the original state. It is a coarse measure of how much spectral weight lies at higher modular energies. Note that a large entropy gap does not necessarily force the nearby Fisher zeros to occur. As we show below, even highly entangled states with an extensive entropy gap can remain parametrically far from Fisher-zero pinching.

\subsection{Spectral competition and the extremal state}
\label{sec:spectral-coexistence}

The preceding bounds identify conditions that prevent Fisher zeros from approaching the continuation region. We now turn to the opposite question: what spectral structure actually produces a nearby zero? The essential
ingredient is competition between sectors of $\rho$ carrying comparable tilted weight. To organize this picture, we characterize the spectrum by its fidelities to two opposite limits. We define
\be
 F_{\mathrm{pure}}(\rho)
    =
    \max_{|\psi\rangle}
    \langle\psi|\rho|\psi\rangle
    =
    \|\rho\|_\infty,
    \label{eq:q_pure}
\ee
and, for a full-rank state on a $d$-dimensional Hilbert space,
\begin{equation}
    F_{\mathrm{mix}}(\rho)
    =
    F\!\left(\rho,\frac{P_\rho}{r}\right)
    =
    \frac{(\Tr\sqrt{\rho})^2}{r},
    \label{eq:x_mixed}
\end{equation}
where $F(\rho, \sigma) = (\Tr\sqrt{\sqrt{\sigma} \rho \sqrt{\sigma}})^2$. We first use these two fidelities to constrain where nearby Fisher zeros can occur. We then turn to the intermediate regime where the strongest spectral competition can arise, and identify a two-sector spectrum that exactly saturates the universal bound of Theorem~\ref{thm:universal_zero_free_wedge}.

\paragraph{Pure-state proximity.} 
If $F_{\mathrm{pure}}(\rho) > 1/2$, one eigenvalue already carries more than half of the weight in $\rho$, and tilting with $u \ge 1$ further enhances this dominance. All the remaining levels will be too weak to cancel it regardless of their phases. This constraints the Fisher zeros in the following way:

\begin{theorem}[Pure-state-fidelity bound]
\label{thm:pure_fidelity_confinement}
Let $\rho$ have rank $r\geq2$ and let
$q=F_{\mathrm{pure}}(\rho)$. If $q>1/2$, every Fisher zero
$z=u+iv$ satisfies
\begin{equation}
    u
    \leq
    u_{\max}(q,r)
    \equiv
    \frac{\log(r-1)}
    {\log\!\left[(r-1)q/(1-q)\right]}
    <1.
    \label{eq:pure_fidelity_confinement}
\end{equation}
\end{theorem}

Thus, once the largest eigenvalue exceeds $1/2$, the entire continuation half-plane $u \geq1$ is free of R\'enyi singularities. The Proof of Theorem 4 can be found in Appendix~\ref{app:pure_fidelity_proof}.

\paragraph{Mixed-state proximity.}

At the opposite spectral limit, $F_{\mathrm{mix}}(\rho)$ quantifies proximity to a flat spectrum. A rigorous finite-dimensional bound is derived in Appendix~\ref{app:mixed_fidelity_proof}. Writing 
$x=F_{\mathrm{mix}}(\rho)$, its large-$r$ form is
\begin{equation}
    \Delta_{\mathrm F}(\rho)
    \gtrsim
    \frac{\pi}{\log[r\,x(1-x)]},
    \qquad
    r\to\infty,
    \label{eq:mixed_fidelity_asymptotic_x}
\end{equation}
for {\it fixed} $0<x<1$. This bound is the weakest at $x=1/2$, indicating that intermediate mixedness is the least protected regime. At $x = 1$ the state is maximally mixed and becomes more protected again with $\Delta_F (\rho) = \infty$.

\paragraph{Two-sector competition and the extremal state.}
The most interesting case occurs in the intermediate regime where neither a single eigenvalue nor a nearly flat spectrum dominates. A simple realization is a spectrum containing one distinguished level and one degenerate sector. Let $|\psi_\star\rangle$ be the distinguished eigenstate and $P_\star=|\psi_\star\rangle\langle\psi_\star|$ its projector,  
and define
\begin{equation}
    \rho(q)
    =
    qP_\star
    +
    \frac{1-q}{r-1}
    (P_\rho-P_\star),
    \qquad
    0<q<1.
    \label{eq:two_sector_family}
\end{equation}
where $q = F_{\mathrm{pure}}(\rho)$ as before, and $P_\rho$ is the projector onto the support of $\rho$. The isolated level carries total weight $q$, and the remaining probability is distributed uniformly over $r-1$ states. Its Fisher zeros are readily obtained as
\begin{equation}
    z_k
    =
    \frac{
        \log(r-1)+i(2k+1)\pi
    }{
        \log[(r-1)q/(1-q)]
    },
    \qquad
    k\in\mathbb Z.
    \label{eq:two_sector_zeros}
\end{equation}

This family makes the origin of a R\'enyi phase transition transparent. For $q>1/r$, the isolated level has the lower modular energy, but the degenerate sector gains an entropic advantage from its multiplicity. Changing $u$ therefore tunes a competition between energy and entropy. The two sectors carry equal tilted weight at
\begin{equation}
    u_c
    =
    \frac{\log(r-1)}
         {\log[(r-1)q/(1-q)]}
    \label{eq:two_sector_uc}
\end{equation}
and the first Fisher zero occurs when they acquire a relative phase of $\pi$.

Equivalently, a transition can be realized at any prescribed $u_c>0$ by choosing 
\begin{equation}
    q=\frac{1}{1+(r-1)^{\,1-1/u_c}}.
    \label{eq:q_for_uc}
\end{equation}
Following~\eqref{eq:two_sector_uc}, the zeros can then be written as 
\begin{equation}
    z_k = u_c + i\,\frac{(2k+1)\pi u_c}
    {\log(r-1)}.
\label{eq:two_sector_zeros_uc}
\end{equation}
If $\log r\propto V$ for a thermodynamic volume $V$, the leading zeros approach the real axis as $1/V$, while the dominant modular sector switches discontinuously across $u_c$. The two-sector family therefore gives an exact realization of a first-order R\'enyi phase transition. From~\eqref{eq:q_for_uc}, the dangerous scenario of $1 \le u_c < 2$ occurs when
\be
\frac{1}{1 + \sqrt{r-1}} < q \le \frac{1}{2}
\ee
At $q=1/2$ and $u_c=1$, define 
\be \label{eq:critical two sector state definition}
\rho_\star\equiv\rho(1/2) = \frac{1}{2} P_{\star} + \frac{1}{2} \frac{P_{\rho} - P_{\star}}{r-1},
\ee
the closest Fisher zero is then
\begin{equation}
    z_\star
    =
    1+\frac{i\pi}{\log(r-1)}.
    \label{eq:critical_zero}
\end{equation}
Thus $\rho_\star$ saturates
Theorem~\ref{thm:universal_zero_free_wedge} and realizes the exact worst-case Fisher-zero clearance at fixed rank $r$.

Fig~\ref{fig:fidelity_bounds} summarizes this structure in the joint-fidelity plane. The range of $(F_{\mathrm{mix}}, F_{\mathrm{pure}})$ values realizable by density matrices is bounded above by spectra in which the remaining weight $1-q$ is distributed uniformly over the other $r-1$ levels, and below by spectra where it is concentrated onto as few levels as possible. In the large-$r$ limit, these boundaries approach $F_{\mathrm{mix}} + F_{\mathrm{pure}}=1$ and the coordinate axes, respectively. The two-sector family lies on the upper boundary, with the extremal state $\rho_{\star}$ at
\begin{equation}
    (x_\star,q_\star)
    =
    \left(
        \frac12+\frac{\sqrt{r-1}}{r},
        \frac12
    \right)
    \xrightarrow{r\to\infty}
    \left(\frac12,\frac12\right).
    \label{eq:fidelity_star}
\end{equation}
where $x_{\star} = F_\text{mix}(\rho_{\star})$. This is the intermediate regime where neither a dominant eigenvalue nor an approximately flat spectrum protects the continuation. The grey area above $q=1/2$ is the zero-free region guaranteed by  Theorem~\ref{thm:pure_fidelity_confinement}. Nearby zeros become possible for $q \le 1/2$. Also shown is the smallest Fisher-zero clearance found numerically among sampled spectra in each $(x,q)$ bin, as indicated by the color scale. The numerical results are consistent with the analytic bounds, with the smallest clearances occurring near the extremal state $\rho_{\star}$.

\subsection{Many-body states}
\label{sec:entanglement-scaling-fisher}

We now examine how the general results above play out in several familiar classes of many-body states. We take $\rho=\rho_A$, the reduced density matrix of subsystem $A$, and illustrate how different structures of the entanglement spectrum can lead to very different Fisher-zero behavior.

\paragraph{Factorized and Gaussian spectra.}

Consider a factorized reduced density matrix,
\begin{equation}
    \rho_A
    =
    \bigotimes_{a=1}^{M}\rho_a,
    \qquad
    f_{\rho_A}(z)
    =
    \prod_{a=1}^{M}f_{\rho_a}(z).
    \label{eq:tensor_product_moment}
\end{equation}
The entropy is additive, whereas the Fisher-zero set is the union of those of the individual factors. Thus extensive entanglement does not by itself lead to Fisher-zero pinching: for example, repeated copies of a fixed factor have entropy proportional to $M$ while their Fisher-zero locations remain unchanged.

Free-fermion Gaussian states provide a natural many-body realization of this structure~\cite{PeschelEisler2009}. Their reduced density matrices factorize into independent fermionic modes
whose Fisher zeros lie on the imaginary axis. This provides an example where extensive entanglement coexists with an entire right half-plane free of R\'enyi singularities.

\paragraph{Area-law states.}

Theorem~\ref{thm:entropy_gap_clearance} gives a general constraint for area-law states.

\begin{corollary}[Area-law states]
\label{cor:area_law_clearance}
If $S_{\mathrm{vN}}(\rho_A) \leq a|\partial A|+b$ where $\partial A$ denotes the boundary of $A$, then
\begin{equation}
    \Re z\geq1,
    \qquad
    |\Im z|
    <
    \frac{\pi}{
        2(a|\partial A|+b)
    }
    \label{eq:area_law_strip}
\end{equation}
is zero-free.
\end{corollary}

An area law therefore bounds how rapidly Fisher zeros can approach the real axis. 
In one dimension,
$|\partial A|=O(1)$ for an interval, so an area law guarantees an $O(1)$ zero-free strip. For example, for a matrix product state of fixed bond dimension $\chi$, a single cut obeys $S_{\mathrm{vN}}\leq\log\chi$ and hence
\begin{equation}
    \Delta_{\mathrm F}
    \geq
    \frac{\pi}{2\log\chi}.
\end{equation}

Volume-law entanglement, in contrast, is compatible with either behavior. On one hand, the two-sector family gives the worst-case example: when $\log r\propto V_A$, its leading Fisher zeros approach the real axis as $O(V_A^{-1})$. On the other hand, factorized spectra can exhibit volume-law entanglement without any such Fisher-zero pinching.

\paragraph{Haar-random states.}

Typical highly entangled spectra provide a useful contrast with the two-sector extremal state. Consider a Haar-random pure state on
$\mathbb C^{d_A}\otimes\mathbb C^{d_B}$ with $d_A\leq d_B$, keeping $d_A / d_B \in(0,1]$ fixed as the dimensions grow. Despite their
volume-law entanglement, these states become increasingly far from Fisher-zero pinching:

\begin{theorem}[Typical Haar clearance]
\label{thm:haar_clearance}
Let $\rho_A$ be the reduced density matrix of a Haar-random pure state with $d_A/d_B \in(0,1]$. Then Fisher-zero clearance 
obeys
\begin{equation}
    \Delta_{\mathrm F}(\rho_A)
    \xrightarrow[d_A\to\infty]{\mathrm{prob.}}
    \infty .
    \label{eq:haar_clearance_diverges}
\end{equation}
Equivalently, for every fixed $T>0$,
\begin{equation}
    \Pr\!\left[
        \Delta_{\mathrm F}(\rho_A)\leq T
    \right]
    \longrightarrow0 .
\end{equation}
\end{theorem}
Thus, with increasing subsystem dimension, typical Haar-random states have no Fisher zeros within any fixed distance of the real axis in the continuation half plane. Proof of the theorem can be found in Appendix~\ref{app:haar_clearance}. We want to highlight that Ref.~\cite{cunden2019moments} has already established that the limiting moment function $M_c(z) $ (see Appendix~\ref{app:haar_clearance} for the definition) has no zeros in the continuation half-plane $\text{Re}(z)\geq 1$. However, they leave open the possibility of zeros escaping to
arbitrarily large real parts while remaining at a bounded
distance from the real axis. Theorem~\ref{thm:haar_clearance} fills in this gap.

Figure~\ref{fig:haar-clearance} shows the corresponding finite-size behavior. Since the most relevant states for Fisher-zero pinching are those in the low-clearance tail, we track the lower 10$\%$ of the clearance distribution over Haar-random reduced states for several $d_A/d_B$ ratios. The clearance increases with $d_A$ in every case over the numerically accessible range. This behavior is in sharp contrast with the two-sector extremal state, whose clearance decreases as $\pi/\log(d_A-1)$.

\begin{figure}[t]
    \centering
    \includegraphics[width=\linewidth]
    {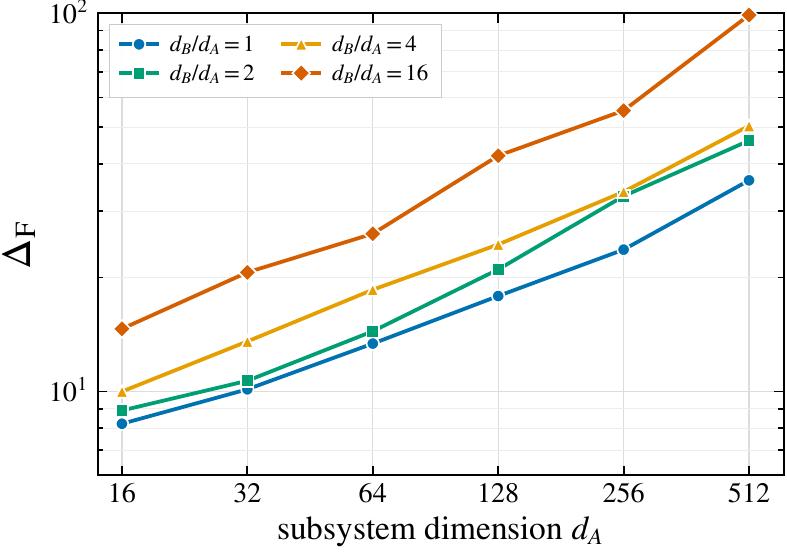}
    \caption{\textbf{Fisher-zero clearance for Haar-random reduced states.}
    For each subsystem dimension $d_A$ and ratio $d_B/d_A$, we show the lower $10\%$ of the empirical $\Delta_{\mathrm F}$ distribution obtained from $150$ Haar-random states. The clearance increases with subsystem size for all ratios shown, in agreement with Theorem~\ref{thm:haar_clearance}.}
    \label{fig:haar-clearance}
\end{figure}

Finally, we turn to how the analytic information above enters the SAC construction. SAC requires choosing a computational domain to be mapped to the unit disc. In practice, we choose a constant-width strip rather than the universal wedge of Theorem~\ref{thm:universal_zero_free_wedge}. This choice is motivated by two considerations: First, the state-dependent examples show that the actual analytic domain can be substantially larger than the conservative universal wedge. Second, the geometry of the domain affects the SAC performance after conformal mapping: for the finite-data continuation considered here, the strip geometry provides a more sensitive and accurate reconstruction than the wedge geometry. The origin of this difference is discussed in detail in Appendix~\ref{app:strip-versus-wedge}.

\section{Stabilized Analytic Continuation}
\label{sec:SAC}

With the analytic structure established in Sec.~\ref{sec:domain_analyticity} and the constant-width strip chosen as the computational domain, we now turn to reconstructing the von Neumann entropy from finitely many integer R\'enyi entropies.
We adapt the stabilized analytic continuation framework of Ciulli and
Spearman~\cite{CS1,CS2,CS3,CS4,CS5} to this problem~\cite{vijay2026analytically}. In the noiseless setting, the construction reduces to a closed-form estimator that is linear in the input R\'enyi entropies, with coefficients fixed entirely by the sampling points and chosen computational domain. Sec.~\ref{sec:the methodology of SAC} gives a brief introduction of the SAC method. The conformal map used to implement the continuation is described in Sec.~\ref{subsec:Two-step conformal mapping}.

\begin{figure*}[t]
\includegraphics[width=1\linewidth]{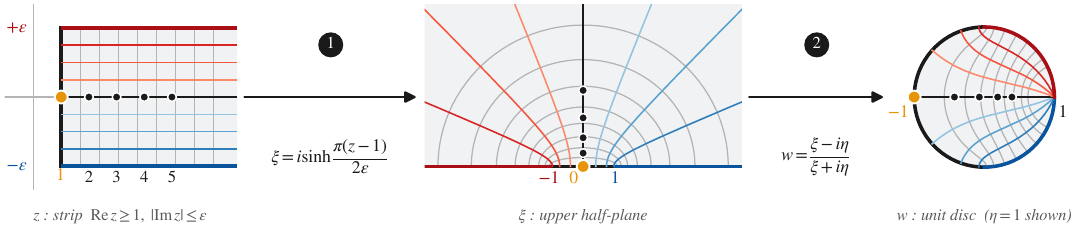}
\caption{\textbf{Two-step conformal transformation.}
The semi-infinite strip $\{\Re z\geq1,\ |\Im z|\leq\epsilon\}$ is mapped to the upper half-plane by $\xi=i\sinh[\pi(z-1)/(2\epsilon)]$, and subsequently to the unit disc by $w=(\xi-i\eta)/(\xi+i\eta)$. The colored curves track a coordinate net through the two transformations. The three boundary segments of the strip map to the unit circle. Black dots denote the integer R\'enyi inputs $z=n$, while the orange point denotes the von Neumann target point, $z=1\mapsto w=-1$. The figure is drawn with $\eta=1$ for clarity; decreasing $\eta$ moves the mapped integer R\'enyi points toward $w=1$.}
\label{fig:two-step map}
\end{figure*}

\subsection{The Method} \label{sec:the methodology of SAC}
Stabilized analytic continuation (SAC)~\cite{vijay2026analytically,CS1, CS2, CS3, CS4, CS5} selects, from the family of analytic functions consistent with a finite set of data points, the unique continuation that minimizes a prescribed norm. Unlike polynomial or rational extrapolation, it does not assume a finite-dimensional functional form.

Consider an unknown complex function $F(w)$ in the unit disc. Ciulli and Spearman showed that $F(w)$ is uniquely determined if the following conditions hold:
\begin{enumerate}
    \item $F(w)$ is analytic for $|w|<1$ and satisfies $F(w)=\overline{F(\bar w)}$;
    \item $F(w)$ takes prescribed real values $a_i\in \mathbb{R}$ at $N$ points $w_i\in(-1,1)$;
    \item among all analytic functions satisfying the above two constraints, $F$ minimizes the boundary norm
    \be
    \label{eq:normB}
    ||F|| =  \frac{1}{2\pi} \int_0^{2\pi} 
    \left|\frac{d}{d \theta}\mbox{Im} F(e^{i \theta}) \right|^2 d \theta 
    \ee
\end{enumerate} 
The rationale for this choice of norm is discussed in Appendix~\ref{sec:choice of norm}. Since Eq.~\eqref{eq:normB} is invariant under a constant shift, $F(w) \to F(w) + c$, it actually defines a seminorm and constrains only the nonconstant part of the function. It is therefore convenient to work with the subtracted function $\widetilde F(w)=F(w) - F(w_1)$ for which $\widetilde F(w_1)=0$. The Ciulli--Spearman construction then uniquely determines the minimum-norm continuation $\widetilde F(w)$, while the known value $F(w_1)=a_1$ restores the additive constant and hence the full function $F(w)$.

These conditions select a unique analytic continuation consistent with the finite data set. Although the minimization is formally over an infinite-dimensional space of analytic functions, the Ciulli-Spearman construction expresses the minimum norm entirely in terms of the interpolation data:
\be \label{eq:delta2min expression}
\delta^2_{\text{min}} 
= \sum_{i,j = 2}^N  (A^{-1})_{ij}   \left(a_i - a_1 \right) \left(a_j - a_1 \right),
\ee
where the matrix elements of $A$ are defined in Eq.~\eqref{eq:A_matrix_elements}. Thus, the functional minimization reduces to the evaluation of a finite-dimensional quadratic form.
The quantity $\delta_{\mathrm{min}}^2$ measures the minimum boundary structure required to accommodate the data. In particular, data that carry the imprint of a singularity on the boundary of the unit disk generally yield a larger minimum norm than data without such a signature.

We now adapt this construction to extract the von Neumann entropy from a small, noiseless set of integer R\'enyi entropies. The input data are $S_{z_i}(\rho)$ at R\'enyi indices $z_i=2,\ldots,N+1$. In Sec~\ref{subsec:Two-step conformal mapping}, we construct a conformal map $z\mapsto w$ that sends the computational domain with strip geometry to the unit disk, the integer R\'enyi points $\{z_{i}\}$ to $\{w_i\}\subset(-1,1)$, and the von Neumann point $z=1$ to the boundary point $w=-1$. With this mapping in place, we introduce the discrepancy function

\begin{align}\label{eq:Discrepancy function}
    D_{\alpha}(z) = \frac{S_{z}(\rho)}{z-1} - \frac{\alpha}{z-1} 
\end{align}
where $\alpha\in \mathbb{R}$ is a variational parameter. Under the map $z\mapsto w$, the function $D_\alpha(w)\equiv D_\alpha(z(w))$ plays the role of $F(w)$ in the Ciulli--Spearman construction. It is analytic in the open unit disk, satisfies $D_\alpha(w)=\overline{D_\alpha(\bar w)}$, and takes the prescribed values $a_i(\alpha) = \frac{S_{z_i}(\rho)-\alpha}{z_i-1}$ at the data points $w_i$.

The essential observation is that the discrepancy function converts the unknown von Neumann limit $n\to 1$ into a pole-cancellation condition. Expanding about $z=1$ gives $D_\alpha(z) = \frac{S_{\mathrm{vN}}(\rho)-\alpha}{z-1} + O(1)$. Thus, for $\alpha\neq S_{\mathrm{vN}}(\rho)$, $D_\alpha$ develops an artificial simple pole at $z=1$, which is mapped to the boundary point $w=-1$ of the unit disc. The boundary norm in Eq.~\eqref{eq:normB} is sensitive to this singularity and formally diverges unless the residue vanishes. This pole is canceled precisely when $\alpha=S_{\mathrm{vN}}(\rho)$. For each candidate value of $\alpha$, the Ciulli--Spearman construction determines the minimum boundary norm $\delta_{\min}^2(\alpha)$ among analytic functions consistent with the corresponding finite set of discrepancy-function values. Since the interpolation values $a_i(\alpha)$ depend linearly on $\alpha$, Eq.~\eqref{eq:delta2min expression} implies that $\delta_{\min}^2(\alpha)$ is quadratic in $\alpha$. Its minimum can be obtained analytically, yielding the closed-form SAC estimate given in Eq.~\eqref{eq:SAC_estimator}. Thus, in the noiseless setting, the continuation reduces to an explicit estimator built from the finite set of integer R\'enyi entropies. The closed-form estimator also allows its intrinsic properties to be analyzed. Appendix~\ref{sec:RKHS_SAC} formulates SAC as a minimum-norm interpolation problem in a reproducing-kernel Hilbert space, with a norm equivalent to Eq.\eqref{eq:normB}, the interpolation gram matrix same as Eq.\eqref{eq:A_matrix_elements}, and the estimator in Eq.\eqref{eq:SAC_estimator} selects the value of $\alpha$ for which the minimum interpolation norm is smallest.
Appendix~\ref{app:Nstar-params} examines how the conformal-map parameters control data-point clustering, and estimates a practical threshold on the number of R\'enyi entropies used for the chosen set of conformal parameters. For noisy inputs, the uncertainties must instead be incorporated into the SAC optimization, see Ref.~\cite{vijay2026analytically}.

\begin{figure*}[t]
    \centering
    \includegraphics[width=\textwidth]{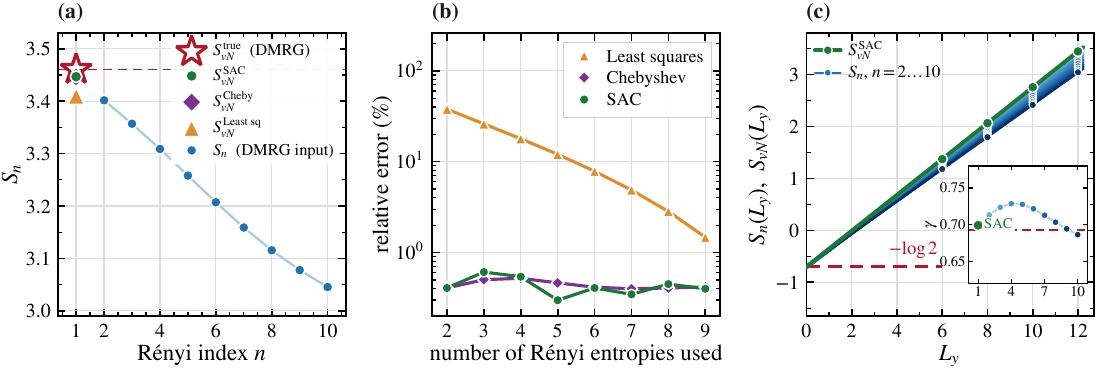}
    \caption{\textbf{Topological entanglement entropy in the toric code.}
    Results for $L_x=48$, $h_x=0.3$, $h_z=0$, and $n_{\max}=10$.
    (a) Integer R\'enyi entropies, the von Neumann entropy directly from DMRG, and estimates from SAC and traditional
    methods for $L_y=12$. The Chebyshev result (purple diamond) nearly overlaps with the DMRG and SAC results and is therefore obscured by the DMRG marker. (b) Relative reconstruction error versus the number of
    R\'enyis used. (c) Finite-circumference scaling of the integer
    R\'enyi entropies and the SAC reconstruction. The latter gives $\gamma_{\mathrm{SAC}}\simeq0.69924$, close to
    $\log2\simeq0.69315$.}
    \label{fig:toric_code}
\end{figure*}
\subsection{Conformal Map}
\label{subsec:Two-step conformal mapping}

To implement SAC, we map the chosen computational domain of strip geometry
\begin{equation}
    \mathcal D_\epsilon
    =
    \left\{
        z:\ \Re z\geq1,\ |\Im z|<\epsilon
    \right\}
\end{equation}
to the unit disc. We use the two-step conformal transformation shown in
Fig.~\ref{fig:two-step map}. First, the semi-infinite strip is mapped to
the upper half-plane by
\begin{equation}
    \xi(z)
    =
    i\sinh\!\left[
        \frac{\pi(z-1)}{2\epsilon}
    \right].
    \label{eq:Schwarz-Christoffel}
\end{equation}
The upper half-plane is then mapped to the unit disc by
\begin{equation}
    w(z)
    =
    \frac{\xi(z)-i\eta}{\xi(z)+i\eta},
    \qquad
    \eta>0.
    \label{eq:second step of conformal map}
\end{equation}
The resulting map sends
\begin{equation}
    z=1\mapsto w=-1,
    \qquad
    z\rightarrow\infty\mapsto w=1,
\end{equation}
while the real R\'enyi axis $z>1$ is mapped to the interval
$-1<w<1$. 

The two parameters have distinct roles. The strip half-width $\epsilon$
specifies the analytic domain assumed in the continuation, whereas
$\eta$ controls the distribution of the mapped R\'enyi points along the
real axis of the disc. If these points cluster too strongly near
$w=1$, the matrix $A$ becomes poorly conditioned. The parameter $\eta$
therefore provides a useful freedom to balance sensitivity to the target
against numerical stability. Equivalently, it parametrizes the remaining
one-parameter family of disc automorphisms that fixes the boundary
points $w=\pm1$.

Throughout this work we use
\begin{equation}
    \epsilon=4,
    \qquad
    \eta=10^{-2},
\end{equation}
which gives reliable reconstructions across the benchmarks considered
below.

The dependence on these parameters and the associated conditioning
tradeoff are discussed in
Appendix~\ref{subsec:choice of conformal parameters}.

\section{Many-body benchmarks}
\label{sec:many_body_benchmarks}

We now turn to three complementary tests of the SAC framework in many-body settings.
First, DMRG provides a stringent test of whether SAC can preserve a subleading topological contribution in the presence of a much larger area-law entropy in two representative spin models. Second, the compactified-boson CFT tests SAC in a quantum-critical state, with independent analytic results available in controlled limits. Finally, conserved dynamics connects the spectral competition identified in Sec.~\ref{sec:spectral-coexistence} to the qualitative separation between ballistic von Neumann growth and subballistic higher-R\'enyi growth. All the numerical benchmarks focus on finite-size and finite-time performance of SAC.

\subsection{Topological Entanglement Entropy in Toric Code and Kagome Heisenberg Model}
\label{sec:dmrg_tee}

Topological entanglement entropy (TEE) is a sensitive benchmark because it appears as a subleading constant in the entropy. DMRG provides both the integer R\'enyi entropies and the von Neumann entropy of the same finite system, allowing us to directly assess the continuation error and compare SAC with conventional extrapolation schemes.
Our numerical setup uses a cylinder geometry with length $L_x=48$ and circumference $L_y$, divided by a cut winding around the periodic direction; see Fig.~\ref{fig:panel_figure}. The entanglement boundary has length $L_y$. For a gapped topological state in the cylinder construction used here, the expected large-circumference behavior is
\begin{equation}
    S_{\mathrm{vN}}(L_y)
    =\alpha L_y-\gamma+\cdots,
    \qquad L_y\to\infty,
    \label{eq:tee_area_law}
\end{equation}
where $\alpha$ is nonuniversal and $\gamma$ is the TEE. The expected value for the $\mathbb Z_2$ topological
phase is $\gamma=\log2$~\cite{TEE_KitaevPreskill,TEE_LevinWen,
zhang2012quasiparticle,jiang2012identifying}.

In principle, $\gamma$ can also be extracted directly from
\begin{equation}
    S_n(L_y)=\alpha_n L_y-\gamma+\delta_n(L_y),
    \qquad n\geq2,
    \label{eq:renyi_tee_scaling}
\end{equation}
where $\delta_n$ denotes the finite-circumference correction. $\gamma$ is expected to be independent of $n$ in the thermodynamic limit, but the corrections at accessible circumferences
can be substantially larger for integer R\'enyi entropies than for the von Neumann entropy~\cite{Flammia_2009,jiang2012identifying,Jiang2013PRL}.
This motivates us to first continue at each circumference and then perform the spatial extrapolation:
\begin{equation}
    \{S_n(L_y)\}_{n=2}^{n_{\max}}
    \xrightarrow{\;\mathrm{SAC}\;}
    S_{\mathrm{vN}}^{\mathrm{SAC}}(L_y)
    \xrightarrow{\;L_y\text{-scaling}\;}
    \gamma_{\mathrm{SAC}}.
    \label{eq:DMRG_route}
\end{equation}

The linearity of SAC gives a useful perspective on this procedure. From Appendix~\ref{subsec:affine-data}, the estimator reproduces constant data exactly, $\sum_n c_n=1$ , where $c_n$'s are as defined in Eq.~\eqref{eq:SAC_estimator}. Applying a fixed set of coefficients to
Eq.~\eqref{eq:renyi_tee_scaling} gives
\begin{equation}
    S_{\mathrm{vN}}^{\mathrm{SAC}}(L_y)
    =\left(\sum_n c_n\alpha_n\right)L_y
     -\gamma+\sum_n c_n\delta_n(L_y).
    \label{eq:tee_sac_decomposition}
\end{equation}
Thus, if the topological constant $\gamma$ is truly independent of the R\'enyi index, SAC preserves it exactly. The only contamination comes
from $\alpha_n$ and $\delta_n(L_y)$. This makes it useful to compare
$S_{\mathrm{vN}}^{\mathrm{SAC}}(L_y)$ directly with the DMRG von Neumann entropy at each circumference before performing the final
$L_y$ extrapolation.

In both models, we use R\'enyi inputs up to $n_{\max}=10$ and compare SAC with least-squares and Chebyshev extrapolations.\footnote{We also tested
Pad\'e-type rational approximants~\cite{Headrick2014pade}. Their
performance depends sensitively on the rational order and on the treatment of spurious poles near $n=1$, so we omit them from the baseline comparisons.}

\begin{figure*}[t]
    \centering
    \includegraphics[width=\textwidth]{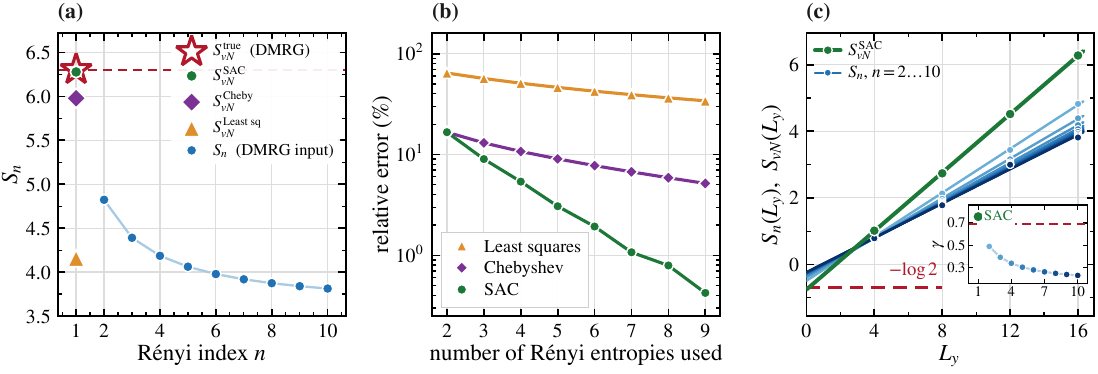}
    \caption{\textbf{Topological entanglement entropy in the Kagome
    Heisenberg model.} Results for $L_x=48$, $J_1=1$, $J_2=0.1$, and
    $n_{\max}=10$. (a) Integer R\'enyi inputs at $L_y=16$, the DMRG
    von Neumann entropy, and the continuation estimates.
    (b) Relative reconstruction error versus the number of integer
    inputs. (c) Finite-circumference scaling. The SAC fit gives
    $\gamma_{\mathrm{SAC}}\simeq0.75602$, compared with
    $\gamma_{\mathrm{DMRG}}\simeq0.69903$ from the directly computed
    von Neumann entropies.}
    \label{fig:kagome}
\end{figure*}

\subsubsection{Toric code}
\label{subsec:toric_code_tee}

We first consider the toric code in a uniform magnetic field,
\begin{equation}
    H=-J_v\sum_v A_v-J_p\sum_p B_p
      -h_x\sum_i\sigma_i^x-h_z\sum_i\sigma_i^z,
    \label{eq:toric_hamiltonian}
\end{equation}
where the spins lie on edges of the square lattice,
$A_v=\prod_{i\ni v}\sigma_i^x$ is the vertex stabilizer, and $B_p=\prod_{i\in\partial p}\sigma_i^z$ is the plaquette stabilizer.
We take $J_v=J_p=1$, $h_x=0.3$, and $h_z=0$, within the $\mathbb Z_2$ topological phase~\cite{Trebst2007loopgas,Vidal2009toriccode}.

Figure~\ref{fig:toric_code}(a) shows the integer R\'enyi entropies at $L_y=12$, the directly computed DMRG von Neumann entropy, and the
continuation estimates. The replica dependence is nearly affine (approximately linear in the R\'enyi index), and both SAC and Chebyshev extrapolation perform well. 
Figure~\ref{fig:toric_code}(b) compares the errors of different continuation methods with varying numbers of input R\'enyis. Applying SAC at each circumference and fitting
Eq.~\eqref{eq:tee_area_law} gives
\begin{equation}
    \gamma_{\mathrm{SAC}}\simeq0.69924,
\end{equation}
close to $\log2\simeq0.69315$; see Fig.~\ref{fig:toric_code}(c).

There are two reasons why SAC performs exceptionally well for the toric code. First, correlations in toric code are extremely short-ranged, with a correlation length likely of order half a lattice spacing. As a result, even a relatively small system, such as $L_y = 6$, is already large enough for finite-size effects to be negligible. This also explains why directly extracting $\gamma$ using R\'enyis works well, as shown in Fig.~\ref{fig:toric_code}(c). Second, the R\'enyi entropies are nearly affine in the R\'enyi index, and SAC is exact for an affine R\'enyi function, as shown in Appendix~\ref{subsec:affine-data}. Any remaining error comes only from the deviation from linearity and inaccuracies in the input R\'enyi data. These favorable conditions do not hold for the Kagome Heisenberg model, which we consider next.

\subsubsection{Kagome Heisenberg model}
\label{sec:kagome_tee}

A more demanding test is provided by the spin-$1/2$ Kagome antiferromagnet,
\begin{equation}
    H=J_1\sum_{\langle i,j\rangle}\mathbf S_i\cdot\mathbf S_j
      +J_2\sum_{\langle\!\langle i,j\rangle\!\rangle}
       \mathbf S_i\cdot\mathbf S_j,
    \label{eq:kagome_hamiltonian}
\end{equation}
with $J_1=1$ and $J_2=0.1$. The sums run over nearest- and next-nearest-neighbor pairs, respectively. This frustrated model has been extensively studied as a candidate quantum spin liquid and also proposed as a minimal model for the Herbertsmithite materials
~\cite{White2011Science,jiang2012identifying,Schollwock2012PRL,
Iqbal2015VMC,He2017Dirac}. Although its thermodynamic TEE is not known from an exact solution, DMRG again supplies an independent reference for the finite-size continuation.

As shown in Fig.~\ref{fig:kagome}(a), the R\'enyi entropies at $L_y=16$ have substantially stronger curvature than in the toric-code example.
The distinction between reconstruction methods is correspondingly sharper. As more R\'enyi entropies are included, SAC improves and substantially outperforms the least-squares and Chebyshev estimates, as seen in Fig.~\ref{fig:kagome}(b). The circumference fit gives
\begin{equation}
    \gamma_{\mathrm{SAC}}\simeq0.75602,
    \qquad
    \gamma_{\mathrm{DMRG}}\simeq0.69903,
\end{equation}
where the second value is obtained by applying the same fit directly to the DMRG von Neumann entropies. 
Also shown in Fig.~\ref{fig:kagome}(c) is the direct fits using integer R\'enyi entropies, which yields an approximately $40\%$ discrepancy in the extracted TEE, consistent with the observations of~\cite{Jiang2013PRL}. This can be attributed to the much longer correlation length in Kagome Heisenberg model, which makes the finite-size effects more severe and causes the direct R\'enyi-based extrapolation to fail.

Finally, we comment on the computational domain used in these DMRG benchmarks. Corollary~\ref{cor:area_law_clearance} states that the clearance is at least of order
$L_y^{-1}$ for area-law states, but does not require the actual clearance to decrease with circumference. Our numerical implementation uses a fixed strip of half-width $\epsilon=4$ for all $L_y$'s. The agreement with DMRG results motivates the possibility that the relevant analytic domains of these states are substantially larger than the worst-case bounds. Although the numerical agreement alone does not establish
the Fisher-zero structure in the thermodynamic limit, the results support a favorable picture that the conservative bounds of
Sec.~\ref{sec:domain_analyticity} do not necessarily set the practical limits
of SAC, and indicate the absence of a R\'enyi transition.

\subsection{$(1+1)d$ Compactified Boson CFT}
\label{sec:CFT}

We next test SAC in a continuum, quantum-critical setting by considering the entanglement entropy and mutual information between two disjoint intervals in a $(1+1)$-dimensional CFT. The mutual
information between two regions $A$ and $B$ is
\begin{equation}
    I(A:B) = S_{\mathrm{vN}}(\rho_A) + S_{\mathrm{vN}}(\rho_B) - S_{\mathrm{vN}}(\rho_{A\cup B}),
\end{equation}
and its R\'enyi counterpart is
\begin{equation}
    I_n(A:B) = S_n(\rho_A) + S_n(\rho_B)- S_n(\rho_{A\cup B}).
    \label{eq:renyiMI}
\end{equation}
For a single interval, the von Neumann and R\'enyi entropies are fixed by conformal symmetries~\cite{Cardy2004QFT, calabrese2009entanglement}, whereas the two-interval entropies depend on the full operator content of the theory~\cite{calabrese2009entanglement, Caraglio2008Twist, Tonni2009disjoint, Tonni2011disjoint}. We focus on the compactified boson, where the integer R\'enyi results are known explicitly but the continuation to the von Neumann limit remains nontrivial~\cite{Hoker2021Taylor, Hubeny2008Holographic, Casini2009Remarks, Rajabpour2012minimal,De_Nobili_2015}. 
We use SAC to reconstruct $I(A:B)$ from a finite set of integer R\'enyi mutual informations and benchmark the estimates against known analytic results in the decompactification and small-cross-ratio limits. In particular, we use the same fixed-width SAC prescription as in the DMRG benchmarks, with strip half-width $\epsilon=4$. As shown below, the SAC estimates agree closely with the available analytic results in both limits, indicating that this fixed computational domain remains effective in a quantum-critical setting, and illustrating that physical criticality does not need to coincide with a R\'enyi phase transition.

\subsubsection{Integer-replica input}

Let $A=[u_1,v_1]$ and $B=[u_2,v_2]$, with $u_1<v_1<u_2<v_2$. Their geometry is characterized by the cross ratio
\begin{equation}
    x=\frac{(v_1-u_1)(v_2-u_2)}{(u_2-u_1)(v_2-v_1)},
    \qquad 0<x<1.
    \label{eq:cft_cross_ratio}
\end{equation}
Small $x$ corresponds to widely separated intervals, while $x\to1$ corresponds to nearly adjacent intervals. The R\'enyi mutual information
is
\begin{equation}
    I_n(A:B)
    =-\frac{c}{6}\left(1+\frac1n\right)\log(1-x)
      +\frac{\log\mathcal F_n(x)}{n-1},
    \label{eq:renyi_mi_combination}
\end{equation}
where $c=1$ for the compactified boson. Unlike the entropy terms, the ultraviolet cutoff and nonuniversal additive constants cancel in this combination. $\mathcal F_n$ contains the additional theory-dependent information absent from the
single-interval problem~\cite{Cardy2004QFT,calabrese2009entanglement,
Tonni2009disjoint}.
\begin{figure}[t]
    \centering
    \includegraphics[width=1\linewidth]{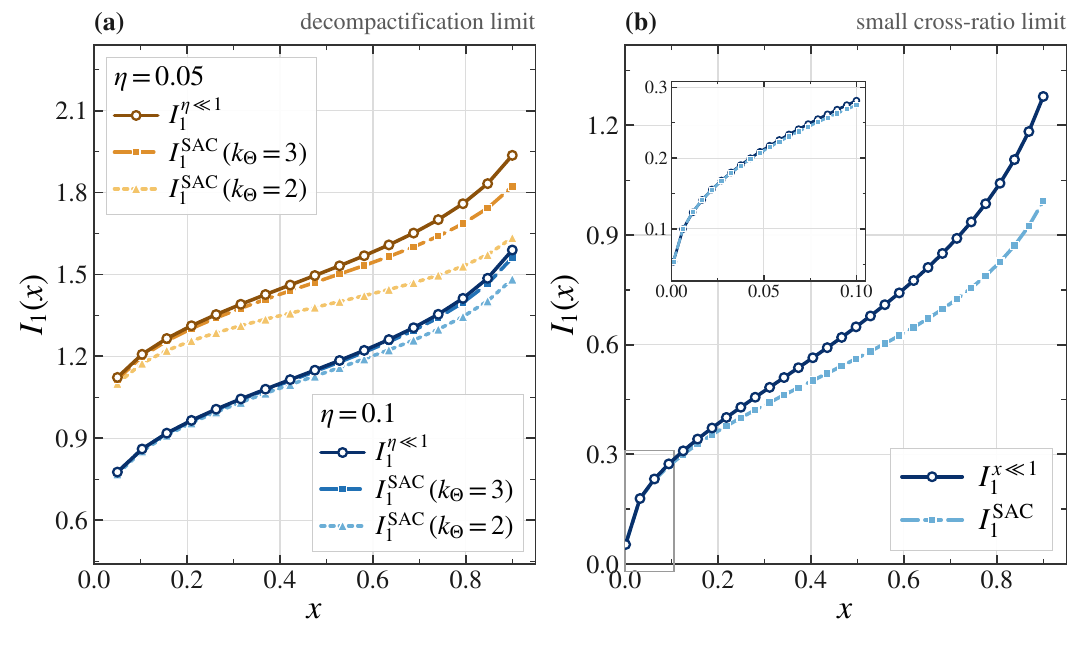}
    \caption{\textbf{Mutual information in the compactified-boson CFT.}
    (a) SAC reconstruction from $I_2,\ldots,I_7$ at $\eta=0.05$ and $0.1$, compared with the decompactification result in Eq.~\eqref{eq:decompactified_MI_exact}.
    The parameter $k_\Theta$ controls the truncation of the theta-function sum; increasing $k_\Theta$ improves the accuracy of the integer-R\'enyi inputs.
    (b) SAC reconstruction at $\eta=3$ and $k_\Theta=2$, compared with the small-cross-ratio expansion in Eq.~\eqref{eq:small_x_mi}. The inset highlights the small-$x$ regime. Both analytic reference expressions are valid only in their respective asymptotic limits.}
    \label{fig:SAC_CFT}
\end{figure}
For compactified boson CFT, it is known in terms of Riemann--Siegel theta functions~\cite{Tonni2009disjoint,Tonni2011disjoint}
\begin{equation}
    \mathcal F_n(x;\eta)
    =\frac{\Theta(0|\eta\Gamma)\Theta(0|\Gamma/\eta)}
           {\Theta(0|\Gamma)^2}.
    \label{eq:compact_boson_Fn}
\end{equation}
Here $\eta$ is the compactification parameter following CFT convention, and should not be confused with the M\"obius parameter in the SAC two-step mapping. The $(n-1)\times(n-1)$ matrix
$\Gamma$ has entries
\begin{equation}
    \Gamma_{rs}
    =\frac{2i}{n}\sum_{k=1}^{n-1}
      \sin\!\left(\frac{\pi k}{n}\right)
      \beta_{k/n}(x)
      \cos\!\left(\frac{2\pi k(r-s)}{n}\right),
    \label{eq:Gamma_matrix}
\end{equation}
for $r,s=1,\ldots,n-1$, where
\begin{equation}
    \beta_y(x)
    =\frac{{}_2F_1(y,1-y;1;1-x)}{{}_2F_1(y,1-y;1;x)}.
\end{equation}
where $_2F_1$ is the hypergeometric function. These expressions provide the integer-replica data. They do not directly provide the continuation to $n = 1$ as the dimension of $\Gamma$ itself is $(n-1) \times (n-1)$.

\subsubsection{SAC for mutual information}

We reconstruct the ordinary mutual information
$I_1=I(A:B)$ directly from the integer values $I_n$. The discrepancy function is
\begin{equation}
    D_\alpha(z)=\frac{I_z-\alpha}{z-1},
    \label{eq:cft_discrepancy}
\end{equation}
where $I_z$ is defined by replacing $n$ with $z$ in~\eqref{eq:renyiMI}. The construction of Sec.~\ref{sec:the methodology of SAC} then gives
\begin{equation}
    I_1^{\mathrm{SAC}}=\sum_{n=2}^{n_{\max}}c_n I_n,
    \label{eq:cft_sac_estimator}
\end{equation}
with the same coefficients as in Eq.~\eqref{eq:SAC_estimator}. This is equivalent to applying SAC to all three constituent entropies and taking their
combination. However, working directly with $I_z$ makes the cancellation of cutoff-dependent terms explicit before continuation.

The analytic structure of $I_z$ is determined by the ratio of three moment functions:
\begin{equation}
    I_z(A:B)
    =\frac{1}{1-z}\log\!\left[
      \frac{f_{\rho_A}(z)f_{\rho_B}(z)}{f_{\rho_{A\cup B}}(z)}
      \right],
    \label{eq:cft_moment_ratio}
\end{equation}
Two observations are useful here. First, a common zero-free domain for all the three moments is sufficient for analyticity of $I_z$, but it is not necessary because zeros can cancel in the ratios.
Second, the finite-rank bounds introduced in Sec~\ref{sec:domain_analyticity} apply only at a finite UV cutoff, and the guaranteed zero-free regions may shrink as we take the continuum limit. We use these bounds only as a guide and benchmark the continuation directly against independent analytic results in controlled limits.

\subsubsection{Benchmarks in controlled limits}

We test SAC against two independent analytic results: the decompactification limit, which provides a reference over a broad range of interval separations, and the small-cross-ratio limit, which tests the recovery of weak correlations between widely separated intervals. In both cases, we use the integer inputs $I_2,\ldots,I_7$ and the same fixed-width SAC map as in the previous benchmarks.

\paragraph{Decompactification limit.}

For a small compactification parameter $\eta$, the von Neumann mutual information has the asymptotic form~\cite{Tonni2009disjoint}
\begin{equation}
    I_1^{\eta\ll1}(x)
    =
    -\frac{c}{3}\log(1-x)
    -\frac12\log\eta
    +\mathcal T(x),
    \label{eq:decompactified_MI_exact}
\end{equation}
where
\begin{align}
    \mathcal T(x)
    &=
    \frac12\bigl[\mathcal J(x)+\mathcal J(1-x)\bigr],
    \\
    \mathcal J(x)
    &=
    \int_{-\infty}^{\infty}dy\,
    \frac{i\pi y}{\sinh^2(\pi y)}
    \log\!\left[{}_2F_1(iy,1-iy;1;x)\right].
\end{align}

We generate the integer-R\'enyi inputs from Eqs.~\eqref{eq:renyi_mi_combination} and~\eqref{eq:compact_boson_Fn}, then apply SAC to reconstruct $I_1(x)$. Fig~\ref{fig:SAC_CFT}(a) shows the results for $\eta=0.05$ and $0.1$.
Near $x=1$ or for very small $\eta$, it becomes more difficult to evaluate the R\'enyi integer inputs accurately, where the theta-function sums converge more slowly. Increasing the truncation parameter $k_\Theta$ retains more terms and improves the agreement, as shown in the figure. Since the comparison is carried out at small but finite $\eta$, part of the remaining discrepancy may come from corrections to the decompactification expression. This should be distinguished from the error introduced by the SAC continuation itself.

\paragraph{Small cross ratio.}

As $x\to 0$, the mutual information vanishes, so recovering it accurately requires resolving a small signal. In this regime, the analytic expansion is~\cite{Tonni2011disjoint,Hoker2021Taylor}
\begin{equation}
    I_1(x)
    =
    -\frac{c}{3}\log(1-x)
    +
    \frac{\sqrt{\pi}\,\Gamma(1+\alpha)}
         {2^{1+2\alpha}\Gamma(\alpha+3/2)}
    x^\alpha
    +\cdots,
    \label{eq:small_x_mi}
\end{equation}
where $\alpha=\min(\eta,\eta^{-1})$.

Fig.~\ref{fig:SAC_CFT}(b) compares the SAC reconstruction at $\eta=3$ with Eq.~\eqref{eq:small_x_mi}. The two agree excellently at
small $x$.
As $x$ increases, higher-order terms omitted from the expansion become important, so discrepancy between the SAC result and~\eqref{eq:small_x_mi} are expected outside the small-$x$ regime.

\begin{figure}[t]
    \centering
    \includegraphics[width=1\linewidth]{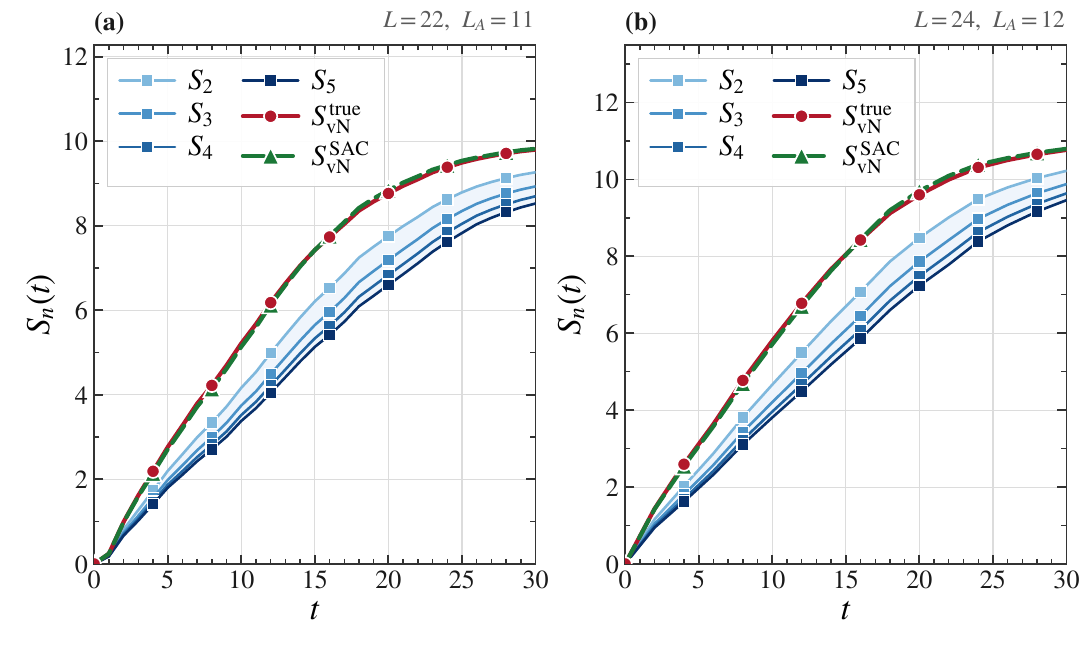}
    \caption{\textbf{Entanglement growth in charge-conserving random circuits.}
    Half-chain entanglement entropies for $U(1)$-symmetric brickwork circuits with (a) $L=22$, $Q=11$, and  $L_A=11$, and (b) $L=24$, $Q=12$, and $L_A=12$.
    Each realization starts from a uniformly sampled computational-basis product state at fixed total charge.
    The curves are averaged over $50$ independent circuit and initial-state realizations for each system size, with $t$ measured in individual circuit layers and entropies expressed in bits.
    SAC reconstruction closely tracks the directly computed von Neumann entropy throughout the displayed time window.}
    \label{fig:SAC_RUCs}
\end{figure}

\subsection{Quantum Dynamics with Conservation Laws}

\label{sec:conservationlaws}

\begin{figure}[t!]
    \centering
    \includegraphics[width=1\linewidth]{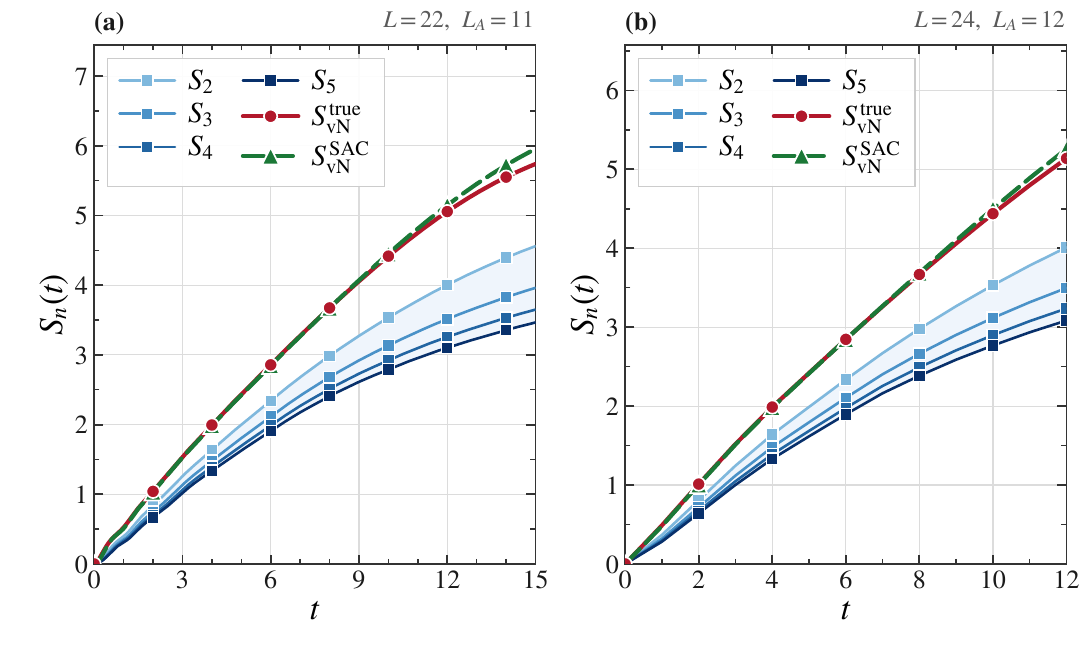}
    \caption{\textbf{Entanglement growth in the tilted-field Ising model.}
    Half-chain entropies for $L=22$ and $24$, with subsystem sizes $L_A=11$ and $12$, respectively, averaged over $50$ random initial states for each system size.
    The SAC reconstruction uses the integer R\'enyi entropies
    $S_2,\ldots,S_{5}$.}
    \label{fig:tilted_ising}
\end{figure}
Conserved dynamics provides a contrasting setting where different R\'enyi indices increasingly probe different parts of the entanglement spectrum. For generic many-body Hamiltonians without conservation laws, both the von Neumann and R\'enyi entropies of a subsystem typically grow linearly in time after a global quench~\cite{EntanglementGrowthRUCs,OperatorSpreadingRUCs,Keyserlingk2018PRX,Zhou2019StatMech}. Diffusion caused by conservation laws can instead produce subballistic higher-R\'enyi growth while the von Neumann entropy remains ballistic~\cite{rakovszky2019sub,Huang2020qudit}. More generally, for $n > 1$ diffusive transport implies
\begin{equation}
    S_{\mathrm{vN}}(t)\sim v_Et,
    \qquad S_n(t)\lesssim O(\sqrt{t \ln t}),
    \label{eq:dynamics_growth_laws}
\end{equation}
where $v_E$ is the entanglement velocity. 

Since the SAC estimator is a linear combination of finitely many R\'enyi entropies, once all the R\'enyis have entered the subballistic regime (up to logarithmic correction), the reconstructed von Neumann entropy must also be upper bounded by subballistic behavior. However, as we show below, SAC accurately recovers the von Neumann entropy in finite-size and finite-time numerics. This apparent tension can be resolved by recognizing that the numerically accessible regime is not asymptotic. At finite-time, higher R\'enyis show more pronounced crossover toward subballistic growth, while lower R\'enyis may still exhibit ballistic growth. Meanwhile, finite-size effects begin to limit the growth of all the entropies, obscuring the asymptotic separation between the two growth laws.

We connect this distinction between finite-time numerics and asymptotic scaling to the analytic structure of the R\'enyi function using a physically motivated two-sector model, building on
Sec.~\ref{sec:spectral-coexistence}. 
In the model, the competition between the two sectors drives Fisher zeros toward $z=1$ and produces a first-order R\'enyi phase transition in the long-time limit. 
The resulting noncommutativity of the long-time and $n\to1$ limits provides a spectral explanation for the different growth laws. It also clarifies how accurate finite-time continuation can coexist with a singular asymptotic limit in which finite-R\'enyi SAC can no longer recover the ballistic von Neumann entropy.

\subsubsection{Random $U(1)$-symmetric circuits}
\label{sec:random_symmetric_circuits}

We first consider periodic chains of $L=22$ and $24$ qubits evolved by a brickwork circuit of independent two-qubit gates drawn from the Haar measure subject to $U(1)$ charge conservation~\cite{rakovszky2019sub}. We work at total charge $Q=L/2$, and sample each initial state uniformly from the computational-basis product states with this total charge. Gates act on alternating sets of nearest-neighbor bonds, with time $t$ counting individual circuit layers.

For each system size $L$, we track the entanglement of a fixed half-chain with $L_A=L/2$, using exact time evolution within the conserved charge sector. We compute the von Neumann entropy and $S_2,\ldots,S_{5}$ after each layer up to $t=30$, and average over 50 independent circuit realizations, each initialized with an independently sampled random initial state. 
Fig.~\ref{fig:SAC_RUCs} show the results. At both sizes, SAC closely follows the directly computed von Neumann entropy throughout the simulated window, including the late-time bending toward saturation. Despite that the finite-size saturation scale shifts with system size, we find the finite-time reconstruction remains accurate.

In~\cite{rakovszky2019sub}, the second Rényi entropy $S_2$ was computed semianalytically for systems as large as $L \sim 100$, by mapping $\Tr(\rho_A^2)$ to an effective partition function and evaluating it using tensor-network contractions. This allows the crossover of $S_2$ from early-time $t$-growth to the asymptotic $\sqrt{t}$ behavior to be clearly resolved. In contrast, our framework requires several integer R\'enyi entropies. As extending the semianalytical construction to $n>2$ rapidly becomes intractable, we instead obtain $S_2, ..., S_5$ by exact time evolution and are limited to $L=24$. Finite-size saturation sets in before the full asymptotic regime, preventing us from cleanly resolving the crossover from $t$ to $\sqrt{t}$ growth particularly for lower R\'enyis. This is the regime where SAC can still accurately reconstruct the von Neumann entropy. Its asymptotic limitation is discussed below in Sec.~\ref{sec:dynamics_replica_interpretation}.

\subsubsection{Tilted-field Ising model}
\label{sec:tilted_field_ising}

To test the continuation in deterministic, energy-conserving dynamics, we next consider the tilted-field Ising chain,
\begin{align}
    H={}&-J\sum_{i=1}^{L-1}\sigma_i^z\sigma_{i+1}^z
       +\sum_{i=1}^{L}(h_z\sigma_i^z+h_x\sigma_i^x)
       \nonumber\\
       &-J(\sigma_1^z+\sigma_L^z),
    \label{eq:tilted_ising_hamiltonian}
\end{align}
with $J=1$, $h_x=(5+\sqrt5)/8$, and $h_z=(1+\sqrt5)/4$, following Ref.~\cite{rakovszky2019sub}. The last term is included to reduce boundary effects.

As in the circuit calculation, we study $L=22$ and $L = 24$ and compute the entropies of a half-chain. Using exact diagonalization, we obtain $S_2,\ldots,S_{10}$ and the von Neumann entropy for each of $50$ random initial states at each size. We apply SAC to each integer-R\'enyi dataset and then average the reconstructed entropies, keeping the SAC parameters fixed across both system sizes.
Fig.~\ref{fig:tilted_ising} shows similarly close agreement between the SAC reconstruction and the directly computed von Neumann entropy for both $L=22$ and $L = 24$. This indicates that the finite-time performance observed in the random-circuit benchmark also extends to Hamiltonian dynamics.

Next, we use the two-sector model to explain how successful reconstruction over the accessible window can coexist with a singular long-time replica limit.

\subsubsection{Two-sector competition and the long-time R\'enyi limit}
\label{sec:dynamics_replica_interpretation}

\begin{figure}[t!]
    \centering
    \includegraphics[width=1\linewidth]{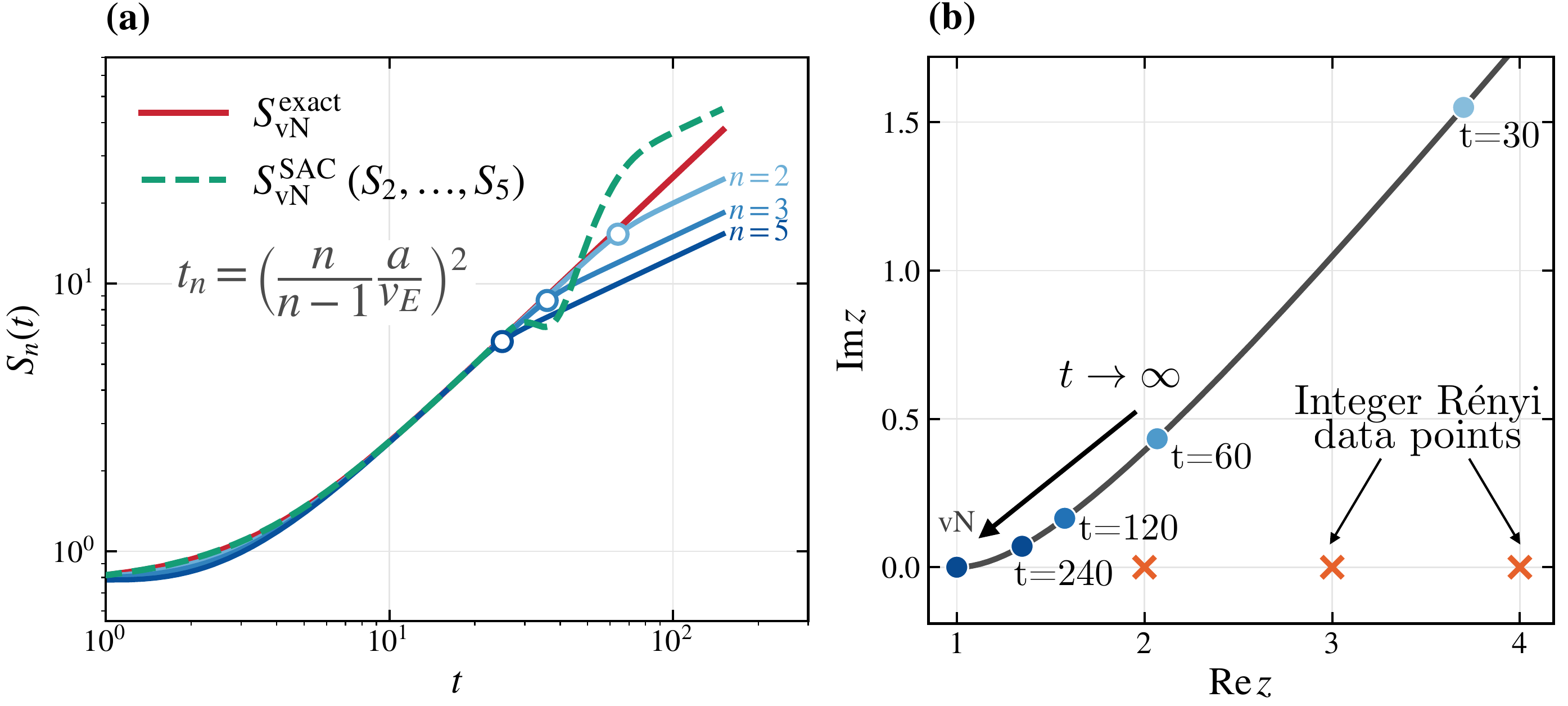}
    \caption{\textbf{Two-sector competition, SAC reconstruction, and Fisher-zero pinching.}
    \textbf{(a)} Exact von Neumann entropy (red), R\'enyi entropies for $n=2,3,5$ (blue), and the noiseless SAC estimate obtained from $S_2,\ldots,S_5$ (green dashed) for the two-sector toy model in Eq.~\eqref{eq:dynamics_two_sector_scaling}, shown on logarithmic axes.
    Open circles mark the estimated crossover times $t_n\simeq[n a/((n-1)v_E)]^2$, defined by $u_*(t_n) = n$, beyond which the isolated eigenvalue dominates the corresponding moment $\operatorname{Tr}[\rho(t)^n]$.
    The higher R\'enyi entropies cross over to $\sqrt{t}$ growth, while the von Neumann entropy grows ballistically. SAC closely follows the exact entropy at early times but develops visible deviations across the crossover.
    \textbf{(b)} Trajectory of the leading Fisher zero $z_0(t)$ of $\operatorname{Tr}[\rho(t)^z]$, corresponding to a branch point of $S_z(\rho(t))$. Blue circles indicate selected times. As time increases, the zero approaches the von Neumann point $z=1$. Orange crosses denote fixed integer R\'enyi data locations.
    This pinching results in a first-order R\'enyi transition in the long-time limit.}
    \label{fig:two_sector_dynamics}
\end{figure}

We now discuss a simple spectral model that makes the connection between the different growth laws and a singular replica limit explicit. 
Motivated by~\cite{rakovszky2019sub, huang2019dynamics}, the key ingredient is competition between one anomalously large eigenvalue and many smaller eigenvalues carrying most of the entropy. Following the two-sector model introduced in Eq.~\eqref{eq:two_sector_family}, we consider a time-dependent state $\rho(t)$, whose spectrum consists of one large eigenvalue $q(t)$ and $R(t)$ degenerate smaller eigenvalues, each equal to $[1-q(t)]/R(t)$. Here $q(t) = F_{\mathrm{pure}}(\rho(t)) = || \rho(t) ||_{\infty}$ is the same $q$ defined in Theorem~\ref{thm:pure_fidelity_confinement}, and $R(t) = r(t) -1$ with $r(t)$ being the rank. We take
\begin{equation}
    q(t)=e^{-a\sqrt t},
    \qquad
    \log R(t)=v_Et+o(t),
    \qquad a,v_E>0.
    \label{eq:dynamics_two_sector_scaling}
\end{equation}
Although $q(t)$ vanishes at long times, it remains much larger than each eigenvalue in the broad sector, since $q(t) \sim e^{-a\sqrt t}$ and $[1-q(t)]/R(t) \sim e^{-v_E t}$. 

It is easy to show that this simple model reproduces the asymptotic separation of ballistic von Neumann growth and subballistic R\'enyi growth. For the von Neumann entropy, the broad sector dominates because its total probability $(1-q(t))$ approaches one and its multiplicity grows as $R(t)\sim e^{v_Et}$, giving
\begin{equation}
    S_{\mathrm{vN}}(t) \sim v_Et.
\end{equation}
For every fixed $n>1$, taking the $n$th power suppresses
its small eigenvalues so strongly that the isolated level eventually dominates the moment, leading to
\begin{equation}
    S_n(t) \sim\frac{n}{n-1}a\sqrt t,
    \qquad n>1.
    \label{eq:dynamics_model_entropies}
\end{equation}
Fig.~\ref{fig:two_sector_dynamics}(a) illustrates the two different behaviors.

\paragraph{Fisher-zero pinching.}
Next, we examine the analytic structure of the model. Consider the moment function
\begin{equation}
    f_\rho(z,t)
    =
    \Tr[\rho(t)^z]
    =
    q(t)^z+R(t)^{1-z}[1-q(t)]^z
    \label{eq:dynamics_two_sector_moment}
\end{equation}
The two terms are the contributions from the isolated eigenvalue and the broad degenerate sector. For real R\'enyi index $u$, their magnitudes become equal at
\begin{equation}
    u_*(t)
    =
    \frac{\log R(t)}
    {\log R(t)+\log[q(t)/(1-q(t))]} .
\end{equation}
Namely, $u_*(t)$ is the crossover R\'enyi index at which the dominant
sector changes.

For complex $z=u+iv$, the two contributions can cancel when they acquire a relative phase of $\pi$ in additional to having equal magnitude. Following Eq~\eqref{eq:two_sector_zeros}, the Fisher zeros are 
\begin{equation}
    z_k(t) =
    \frac{\log R(t)+i(2k+1)\pi}
    {\log R(t)+\log[q(t)/(1-q(t))]},
    \qquad k\in\mathbb Z .
\end{equation}
The leading zero $z_0(t)$ satisfies 
\be
\Re z_0 = u_*(t), \quad \Im z_0 = \frac{\pi\Re z_0}{\log R(t)} =\frac{\pi \Re z_0}{\log (r(t)-1)}
\ee
i.e., the leading zero saturates the universal wedge of Theorem~\ref{thm:universal_zero_free_wedge}. Furthermore, using $q(t)=e^{-a\sqrt t}$ and $\log R(t)=v_E t+o(t)$, we have
\begin{equation}
    \Re z_0 \sim 1+ \frac{a}{v_E\sqrt t}, \quad
    \Im z_0 \sim \frac{\pi}{v_E t}.
    \label{eq:dynamics_zero_asymptotics}
\end{equation}
Therefore, the crossover R\'enyi index approaches the von Neumann point from the right at a rate $t^{-1/2}$, while the zeros pinch the real axis at a rate $t^{-1}$. Fig.~\ref{fig:two_sector_dynamics}(b) illustrates the process.

The first-order character of the transition can be seen from the long-time behavior of the modular free energy:
\begin{equation}
    \lim_{t\to\infty}-\frac{1}{ut}\log f_\rho(u,t)
    =
    \begin{cases}
        v_E \bigg(1-\frac{1}{u}\bigg), & 0<u<1,\\
        0,       & u\geq1.
    \end{cases}
    \label{eq:dynamics_limiting_pressure}
\end{equation}
Here, we rescale by $1/t$ because $\log R(t)\sim v_E t$ causing $\log f_{\rho}(u, t)$ to be extensive in time, with $t$ playing the role of volume in the usual thermodynamic limit. For $u<1$, the broad sector dominates, whereas for $u>1$ the isolated eigenvalue dominates. Consequently, the first derivative with respect to the modular temperature $1/u$ is discontinuous at $u=1$.

The different entropy growth laws are a consequence of this switch.
Taking $n\to1$ first retains the broad sector responsible for ballistic
von Neumann growth,
\be
\lim_{t \to \infty} \lim_{n \to 1} S_n(t) / t = v_E
\ee
whereas taking $t\to\infty$ at any fixed $n>1$ selects the isolated eigenvalue and gives subballistic growth,
\be
 \lim_{n \to 1} \lim_{t \to \infty} S_n(t) / t = 0
\ee
Thus the replica and long-time limits do not commute.

\paragraph{Implications for SAC.}

The two-sector model makes the asymptotic limitation of SAC more explicit. Once $u_*(t)<2$, every integer input lies on the side of the crossover that favors the isolated eigenvalue. 
The broad sector carrying the ballistic entropy is still present in the exact moments, but the accessible replicas become insensitive to it.

This provides a spectral interpretation of the general limitation mentioned at the beginning of Sec.~\ref{sec:conservationlaws}: once all sampled R\'enyi entropies are subballistic, any SAC reconstruction must also be subballistic. There is no contradiction with the successful finite-time reconstructions in Figs.~\ref{fig:SAC_RUCs} and~\ref{fig:tilted_ising}, which probe an intermediate regime where the different R\'enyi entropies are still at different stages of crossover and finite-size saturation begins to set in.

\section{Conclusion}
\label{sec:conclusions}

In this work, we connected the reconstruction of the von Neumann entropy from a finite set of noiseless integer R\'enyi entropies to the analytic structure of the R\'enyi function itself. In particular, we showed that zeros of $\Tr\rho^z$ generate branch-point singularities of $S_z(\rho)$ that can fundamentally obstruct analytic continuation. When these singularities pinch the real $z$ axis in the interval $1\leq z<2$, the von Neumann and integer-R\'enyi entropies are separated by a R\'enyi phase transition, so the $z\to1$ limit is no longer a smooth extrapolation.

Our zero-free results constrain when such an obstruction can occur. We first established a universal, rank-dependent zero-free wedge and then systematically sharpened the exclusion bounds using additional spectral information, including the modular spectral width, the entropy gap, and fidelities to pure and maximally mixed states. 
We further identify a two-sector family that saturates the universal bound and realizes a first-order R\'enyi transition as its rank grows. 
For area-law states satisfying $S_{\mathrm{vN}}(\rho_A)\leq a|\partial A|+b$, with size-independent constants $a$ and $b$, we establish a zero-free strip in $\Re z\geq1$ of half-width $\pi/[2(a|\partial A|+b)]$. 
Most strikingly, for reduced density matrices of Haar-random bipartite states at fixed subsystem-to-environment dimension ratio, we prove that the zero clearance in $\Re z\geq1$ diverges in probability as the subsystem dimension grows.
These results show that proximity to a R\'enyi transition is controlled not simply by the amount of entanglement, but by how spectral weight is organized across the modular spectrum.

We further developed the stabilized analytic continuation (SAC) method to reconstruct the von Neumann entropy from a finite set of noiseless integer-R\'enyi entropies. In this setting, SAC reduces to a closed-form linear estimator with state-independent coefficients.
We benchmarked the method in three many-body settings. Using DMRG data, SAC substantially improves the extraction of the subleading topological entanglement entropy compared to conventional extrapolation schemes.
In the compactified-boson CFT, it reconstructs disjoint-interval mutual information in agreement with independent analytic limits, illustrating that physical criticality need not imply R\'enyi criticality. 
Finally, SAC accurately tracks finite-time von Neumann entropy growth in charge-conserving random circuits and the tilted-field Ising chain. Across these benchmarks, we use a fixed numerical strip wider than the certified zero-free region. Its success supports the practical utility of this regularization choice, without establishing analyticity throughout the assumed strip.

Conserved dynamics also makes the distinction between finite-time success and asymptotic obstruction particularly transparent. Once all sampled R\'enyi entropies become subballistic, any SAC estimator constructed from a fixed finite set of them must itself remain subballistic and therefore cannot recover ballistic von Neumann growth asymptotically. We connect this limitation to the underlying analytic structure through a physically motivated time-dependent two-sector model that reproduces ballistic von Neumann growth and subballistic higher-R\'enyi growth. In this model, Fisher zeros of $\Tr\rho(t)^z$ pinch the real $z$ axis at $z=1$ in the long-time limit, producing a first-order R\'enyi transition. It provides a concrete spectral mechanism by which accurate finite-time reconstruction can coexist with a genuine obstruction to asymptotic analytic continuation.

More broadly, SAC offers a framework for recovering physically relevant quantities from more accessible replica observables. Natural applications include the R\'enyi coherent information in quantum error correction ~\cite{Ehud_Ashwin,Vijay2026} and R\'enyi correlators used to diagnose strong-to-weak spontaneous symmetry breaking
~\cite{LeeJianSWSSB,Olumakinde2023,Strong_Weak_SSB_Open_Systems,
Strong_to_weak_Fidelity_Correlator,Strong_Weak_SSB_Renyi1_Correlator,
Strong_Weak_SSB_Wightman}. In these settings, transition points and critical behavior can strongly depend on the R\'enyi index, so finite-R\'enyi diagnostics need not serve as accurate proxies of the physically relevant limit. 
Other natural targets include logarithmic negativity ~\cite{calabrese2012entanglement,calabrese2013entanglement} and the $(\alpha,z)$ R\'enyi divergences ~\cite{Petz1986quasi,KudlerFlam2024RMI}. Their analytic structures can differ substantially from that of the R\'enyi entropy, so each extension requires a separate analysis of the locations of singularities. The central question nevertheless remains the same: whether the accessible higher-replica observables are analytically connected to the operationally relevant limit, or whether a replica transition intervenes.

A further direction is to adapt SAC to the reconstruction of real-frequency response functions from imaginary-time or Matsubara-frequency data in quantum Monte Carlo~\cite{Silver1990MaxEnt, Jarrell1996Bayesian, Beach2000Pade, Gunnarsson2010optical, Schott2016averaging, Gull2011QMCRMP, LeBlanc2015Simons, Gull2021PRL, Zhang2024minipole, Zhang2025minipole}. In this setting, numerical calculations provide Green's functions away from the real-frequency axis, while experimentally relevant spectral and response functions are encoded in their boundary values on the real axis. The present minimum-norm construction could potentially be adapted by conformally mapping the appropriate complex-frequency domain to the unit disc while incorporating additional physical constraints, such as causality and spectral sum rules.

Finally, a more ambitious direction is to develop a renormalization-group description of R\'enyi phases for modular Hamiltonians arising in many-body states. Such a theory would identify the effective degrees of freedom and couplings that drive R\'enyi transitions, classify their universal behavior, and predict the finite-size scaling of the leading Fisher zeros. This requires retaining not only the modular spectrum but also the spatial structure of the modular Hamiltonian and how it changes under coarse-graining as the R\'enyi index is varied. This would help understand which structures in a many-body state protect a common analytic branch connecting accessible integer R\'enyi points to the von Neumann point, and which instead drive a transition that separates them.

\acknowledgments
The authors thank Nima Lashkari and Tibor Rakovszky for helpful discussions. The authors are also grateful to Tibor Rakovszky for sharing the code of tilted field Ising model. A.V. is supported by Jong Yeon Lee's faculty startup grant at the University of Illinois, Urbana-Champaign, and the IBM-Illinois Discovery Accelerator Institute. ChatGPT and Claude were used to assist with proving some of the theorems establishing the zero-free regions and identifying analytic properties of the SAC estimator.

\bibliography{References}

\appendix
\addtocontents{toc}{\protect\hideappendixsubsections}

\renewcommand{\theequation}{\thesection.\arabic{equation}}

\section{Proofs of the Analytic-Domain Bounds}
\label{app:analyticity_proofs}

In this appendix we prove the zero-free results stated in
Sec.~\ref{sec:domain_analyticity}. Throughout, $\rho$ is a density
matrix of rank $r$, $P_\rho$ denotes its support projector, and all
operator functions are understood on the support of $\rho$. We write
\begin{equation}
    f_\rho(z)=\Tr\rho^z,
    \qquad
    z=u+iv,
\end{equation}
and use the tilted state
\begin{equation}
    \rho_u=\frac{\rho^u}{\Tr\rho^u}.
\end{equation}

\subsection{Universal zero-free wedge}
\label{app:universal_wedge_proof}

 We first prove a lemma that will be invoked in proving
Theorem~\ref{thm:universal_zero_free_wedge}. Let
$q=\|\rho\|_\infty$ and let $P_\star$ denote the projector onto the
maximal eigenspace of $\rho$, with
\begin{equation}
    m\equiv\operatorname{rank}P_\star .
\end{equation}

\begin{lemma}[Logarithmic-spiral bound]
\label{lem:log_spiral}
For $\alpha>0$ and $t\geq0$, define
\begin{equation}
    w_\alpha(t)=e^{-(1+i\alpha)t}.
\end{equation}
Then the real-linear functional
\begin{equation}
    \mathcal L_\alpha(w)
    \equiv
    \Re w-\frac{1}{\alpha}\Im w
\end{equation}
satisfies
\begin{equation}
    \mathcal L_\alpha(w_\alpha(t))
    \geq
    -e^{-\pi/\alpha}.
    \label{eq:spiral_bound}
\end{equation}
\end{lemma}

\begin{proof}
Writing $\phi=\alpha t$,
\begin{equation}
    \mathcal L_\alpha(w_\alpha)
    =
    e^{-\phi/\alpha}
    \left(
        \cos\phi+\frac{1}{\alpha}\sin\phi
    \right).
\end{equation}
Differentiation gives
\begin{equation}
    \frac{d}{d\phi}
    \mathcal L_\alpha(w_\alpha)
    =
    -\left(1+\alpha^{-2}\right)
    e^{-\phi/\alpha}\sin\phi.
\end{equation}
The global minimum therefore occurs at $\phi=\pi$, yielding
Eq.~\eqref{eq:spiral_bound}.
\end{proof}

\begin{proof}[Proof of Theorem~\ref{thm:universal_zero_free_wedge}]
Suppose $f_\rho(u+iv)=0$ with $u>0$ and, without loss of generality,
$v>0$. Let
\begin{equation}
    Q_\star=P_\rho-P_\star,
    \qquad
    R\equiv q^{-1}Q_\star\rho Q_\star.
\end{equation}
The spectrum of $R$ lies strictly inside $(0,1)$ and
\begin{equation}
    0
    =
    q^{-z}f_\rho(z)
    =
    m+\Tr_{Q_\star}R^z.
    \label{eq:wedge_top_sector}
\end{equation}
Set
\begin{equation}
    \alpha=\frac{v}{u},
    \qquad
    T=-u\log R\succeq0.
\end{equation}
Then
\begin{equation}
    R^{u+iv}
    =
    e^{-(1+i\alpha)T}.
\end{equation}
Applying the functional of
Lemma~\ref{lem:log_spiral} to
Eq.~\eqref{eq:wedge_top_sector} gives
\begin{equation}
    0
    \geq
    m-(r-m)e^{-\pi/\alpha}.
\end{equation}
If $m\geq r/2$, this is impossible for finite $\alpha$ (if the maximal eigenspace constitutes for more than half of the rank, the phases in \eqref{eq: zeros of characteristic function} can never cancel out to zero), so no zeros occur
in $\Re z>0$. Otherwise,
\begin{equation}
    \alpha
    \geq
    \frac{\pi}{
        \log[(r-m)/m]
    }
    \geq
    \frac{\pi}{\log(r-1)}.
\end{equation}
Hence every zero satisfies
\begin{equation}
    |v|
    \geq
    \frac{\pi u}{\log(r-1)}.
\end{equation}
For $r=2$, the maximal eigenspace necessarily gives either $m\geq r/2$
or a flat state, so the entire right half-plane is zero-free.

Sharpness follows from the family in
Eq.~\eqref{eq:two_sector_family}, whose zeros satisfy
Eq.~\eqref{eq:two_sector_zeros}. Their slope is
\begin{equation}
    \frac{|\Im z_k|}{\Re z_k}
    =
    \frac{(2k+1)\pi}{\log(r-1)},
\end{equation}
and the lowest branch therefore saturates the universal wedge.
\end{proof}

\subsection{Spectral-width bound}
\label{app:spectral_width_proof}

\begin{proof}[Proof of Theorem~\ref{thm:spectral_width_bound}]
Let
\begin{equation}
    K_{\min}\equiv\lambda_{\min}(K_\rho),
    \qquad
    K_{\max}\equiv\lambda_{\max}(K_\rho),
\end{equation}
so that
\begin{equation}
    \Delta_K=K_{\max}-K_{\min}.
\end{equation}
Choose the midpoint
\begin{equation}
    K_c=\frac{K_{\min}+K_{\max}}{2}.
\end{equation}
For $|v|\Delta_K<\pi$,
\begin{equation}
    \left|
        v(K_\rho-K_c)
    \right|
    <\frac{\pi}{2}
\end{equation}
throughout the spectrum of $K_\rho$. Hence
\begin{equation}
    \cos[v(K_\rho-K_c)]\succ0 .
\end{equation}
For any $u>0$,
\begin{align}
    \Re\!\left[
        e^{ivK_c}f_\rho(u+iv)
    \right]
    &=
    \Tr\!\left[
        \rho^u
        \cos[v(K_\rho-K_c)]
    \right]
    >0.
\end{align}
Thus $f_\rho(u+iv)\neq0$ whenever
$|v|<\pi/\Delta_K$.
\end{proof}

\subsection{Entropy-gap and area-law bounds}
The entropy-gap bound follows from a modular version of the Margolus–Levitin orthogonalization bound, together with monotonicity under spectral tilting~\cite{margolus1998maximum}.

\label{app:entropy_gap_proof}

\begin{lemma}
\label{lem:mean_energy_scalar}
For all $x\geq0$,
\begin{equation}
    x-\frac{\pi}{2}
    +\frac{\pi}{2}\cos x+\sin x
    \geq0.
    \label{eq:mean_energy_scalar}
\end{equation}
\end{lemma}

\begin{proof}
Let the left-hand side be $g(x)$. On $0\leq x\leq\pi$,
\begin{equation}
    g'(x)
    =
    2\cos\frac{x}{2}
    \left[
        \cos\frac{x}{2}
        -
        \frac{\pi}{2}\sin\frac{x}{2}
    \right].
\end{equation}
Thus $g$ first increases and then decreases, while
$g(0)=g(\pi)=0$, so $g\geq0$ on this interval. On
$\pi\leq x\leq2\pi$, both terms inside the square brackets are
nonpositive, but since $\cos\frac{x}{2} \le 0$, hence $g'(x)\geq0$. Finally, for $x\geq2\pi$,
\begin{equation}
    g(x)\geq x-\pi-1>0.
\end{equation}
\end{proof}

\begin{proof}[Proof of Theorem~\ref{thm:entropy_gap_clearance}]
Define
\begin{equation}
    K_{\min}
    =
    \lambda_{\min}(K_\rho)
    =
    S_\infty(\rho),
    \qquad
    X\equiv K_\rho-K_{\min}P_\rho\succeq0.
\end{equation}
At a zero of $f_\rho(u+iv)$,
\begin{equation}
    \Tr(\rho_u e^{-ivX})=0
\end{equation}
up to a nonvanishing phase. Hence, taking $v>0$ 
and applying Lemma~\ref{lem:mean_energy_scalar} to the positive operator
$vX$ at the zeros of $f(u+ iv)$ gives
\begin{equation}
    v\,\Tr(\rho_uX)\geq\frac{\pi}{2}.
    \label{eq:mean_energy_zero}
\end{equation}

Define
\begin{equation}
    \mu(u)=\Tr(\rho_uX).
\end{equation}
Differentiation gives
\begin{equation}
    \mu'(u)
    =
    -\operatorname{Var}_{\rho_u}(K_\rho)
    \leq0.
\end{equation}
Therefore, for $u\geq1$,
\begin{align}
    \mu(u)
    &\leq
    \mu(1)
    \nonumber\\
    &=
    \Tr(\rho K_\rho)-K_{\min}
    \nonumber\\
    &=
    S_{\mathrm{vN}}(\rho)-S_\infty(\rho).
\end{align}
Combining this with Eq.~\eqref{eq:mean_energy_zero} proves
Eq.~\eqref{eq:entropy_gap_clearance}.

If
$S_{\mathrm{vN}}=S_\infty$, then
$\Tr[\rho(K_\rho-K_{\min})]=0$. Since the operator in parentheses is
positive, it vanishes on the support of $\rho$, implying
$K_\rho=K_{\min}P_\rho$ and therefore
$\rho=P_\rho/r$.
\end{proof}

\begin{proof}[Proof of Corollary~\ref{cor:area_law_clearance}]
Since $S_\infty(\rho_A)\geq0$,
\begin{equation}
    S_{\mathrm{vN}}(\rho_A)-S_\infty(\rho_A)
    \leq
    S_{\mathrm{vN}}(\rho_A)
    \leq
    a|\partial A|+b.
\end{equation}
Substitution into
Theorem~\ref{thm:entropy_gap_clearance} proves the result.
\end{proof}

\subsection{Pure-state-fidelity bound and the critical state}
\label{app:pure_fidelity_proof}

\begin{proof}[Proof of Theorem~\ref{thm:pure_fidelity_confinement}]
Let
\begin{equation}
    q=\|\rho\|_\infty>\frac12
\end{equation}
and let
$P_\star=|\psi_\star\rangle\langle\psi_\star|$
project onto a state attaining this maximum. Set
$Q_\star=P_\rho-P_\star$. At a zero of $f_\rho(z)$,
\begin{equation}
    \Tr(P_\star\rho^z)
    =
    -\Tr(Q_\star\rho^z),
\end{equation}
and therefore
\begin{equation}
    q^u
    \leq
    \Tr(Q_\star\rho^u).
    \label{eq:pure_balance}
\end{equation}

For $0<u\leq1$, define
\begin{equation}
    \sigma
    =
    \frac{Q_\star\rho Q_\star}{1-q}.
\end{equation}
The state $\sigma$ has rank at most $r-1$, and the concavity of
$x^u$ gives
\begin{equation}
    \Tr\sigma^u
    \leq
    (r-1)^{1-u}.
\end{equation}
Hence
\begin{equation}
    \Tr(Q_\star\rho^u)
    \leq
    (r-1)^{1-u}(1-q)^u.
\end{equation}
Combining this with Eq.~\eqref{eq:pure_balance} yields
\begin{equation}
    u
    \leq
    \frac{\log(r-1)}
    {\log[(r-1)q/(1-q)]}.
\end{equation}
For $u\geq1$, since $x^u$ is convex, we have
\begin{equation}
    \Tr(Q_* \rho^u) \leq (1-q)^u < q^u
\end{equation}
contradicting Eq.~\eqref{eq:pure_balance}. Thus no zero exists for
$u\geq1$ when $q > \frac{1}{2}$.
\end{proof}

For the two-sector family
Eq.~\eqref{eq:two_sector_family},
\begin{equation}
    \Tr\rho_q^z
    =
    q^z+(r-1)^{1-z}(1-q)^z.
\end{equation}
Setting this expression to zero gives directly
\begin{equation}
    z_k
    =
    \frac{
        \log(r-1)+i(2k+1)\pi
    }{
        \log[(r-1)q/(1-q)]
    }.
\end{equation}
At $q=1/2$ this becomes
\begin{equation}
    z_k
    =
    1+\frac{i(2k+1)\pi}{\log(r-1)}.
\end{equation}
Theorem~\ref{thm:universal_zero_free_wedge} implies that no rank-$r$
state can possess a zero with $\Re z\geq1$ and
$|\Im z|<\pi/\log(r-1)$. Thus the critical state
$\rho_\star$ realizes the exact worst-case clearance in the SAC
half-plane.

\subsection{Maximally mixed-state fidelity bound}
\label{app:mixed_fidelity_proof}

Firstly, define
\begin{equation}
    h = 1 - \sqrt{F_\text{mix}(\rho)}
\end{equation}
where $F_\text{mix}(\rho)$ is defined in Eq.~\eqref{eq:x_mixed}.

\begin{theorem}[Maximally mixed-state fidelity bound]
\label{thm:mixed_fidelity_bound}
Let $\rho$ be a rank $r$ density matrix and set
\begin{equation}
    q_\text{max}(h)
    \equiv
    \min\!\Big\{1,\;\big(h+r^{-1/2}\big)\big(2-h-r^{-1/2}\big)\Big\}.
\end{equation}
Then every zero $z=u+iv$ of $f_\rho(z)=\Tr\rho^z$ with $u\geq1$ obeys
$|\Im z|\geq B^{\mathrm{SAC}}_{\mathrm{mix}}(h)$, where
\begin{equation}\label{eq:mixed_master_width}
    B^{\mathrm{SAC}}_{\mathrm{mix}}(h)
    \equiv
    \max_{0<\theta<\pi/2}
    \frac{\theta}{\log\!\Big[(1-h)\sqrt{r\, q_\text{max}(h)}\;
        \dfrac{1+\cos\theta}{\cos\theta}\Big]},
\end{equation}
the maximization taken over $\theta$ for which the logarithm is positive.
Whenever the balance condition cannot be satisfied, the entire half-plane
$\Re z\geq1$ is zero-free. Consequently $S_z(\rho)$ admits an analytic branch in
$\{\Re z\geq1,\ |\Im z|<B^{\mathrm{SAC}}_{\mathrm{mix}}(h)\}$. For fixed $h$ and
large $r$,
\begin{equation}\label{eq:mixed_master_asymptotic}
    B^{\mathrm{SAC}}_{\mathrm{mix}}(h)
    =\frac{\pi}{\log[\,r\,g(h)\,]}\,[1+o(1)],
    \qquad g(h)=h(2-h)(1-h)^2,
\end{equation}
up to subleading $\log\log r$ corrections.
\end{theorem}

We first establish two elementary ingredients.

\begin{lemma}[Bound on the largest pure-state weight]
\label{lem:qmax_h}
With
\begin{equation}
    h=1-\frac{1}{\sqrt r}\Tr\sqrt\rho,
\end{equation}
one has
\begin{equation}
    \|\rho\|_\infty
    \leq
    q_{\max}(h)
    =
    \min\left\{
        1,\,
        (h+r^{-1/2})(2-h-r^{-1/2})
    \right\}.
\end{equation}
\end{lemma}

\begin{proof}
Let $q=\|\rho\|_\infty$ and let $P_\star$ be a rank-one projector
attaining $q$. By Cauchy--Schwarz,
\begin{align}
    1-h
    &=
    \frac{1}{\sqrt r}\Tr\sqrt\rho
    \nonumber\\
    &\leq
    \frac{1}{\sqrt r}
    +
    \sqrt{1-q}.
\end{align}
Since $\Tr\sqrt\rho\geq1$, one has
$1-h\geq r^{-1/2}$, and therefore
\begin{equation}
    q
    \leq
    1-\left(1-h-r^{-1/2}\right)^2
    =
    (h+r^{-1/2})(2-h-r^{-1/2}).
\end{equation}
The additional minimum with unity is trivial.
\end{proof}

\begin{lemma}[Window cancellation]
\label{lem:window_cancellation}
Choose an entanglement-energy interval
\[
    I=[E_c-L/2,E_c+L/2]
\]
of width $L$, and let $D$ be the spectral projector onto
all eigenstates of the entanglement Hamiltonian whose
energies lie outside $I$. Define
\[
    \delta_u=\Tr(\rho_uD),
\]
so that $\delta_u$ and $1-\delta_u$ are the tilted
probabilities outside and inside $I$, respectively.

If $f_\rho(u+iv)=0$ and
$\theta\equiv |v|L/2<\pi/2$, then
\begin{equation}
    \delta_u
    \geq
    \frac{\cos\theta}{1+\cos\theta}.
    \label{eq:window_cancel}
\end{equation}
\end{lemma}

\begin{proof}
Let $E_j$ and $p_j(u)=\langle j|\rho_u|j\rangle$ denote
the entanglement energies and their tilted probabilities.
Multiplying the zero trace condition
$\sum_j p_j(u)e^{-ivE_j}=0$ by $e^{ivE_c}$ and taking
the real part gives
\[
    0=\sum_{E_j\in I}p_j(u)\cos\!\bigl(v(E_j-E_c)\bigr)
      +\sum_{E_j\notin I}p_j(u)\cos\!\bigl(v(E_j-E_c)\bigr).
\]
Inside $I$, $|v(E_j-E_c)|\leq\theta$, so the first sum
is at least $(1-\delta_u)\cos\theta$. The second sum
is at least $-\delta_u$, since $\cos x\geq-1$.
Thus
\[
    0\geq(1-\delta_u)\cos\theta-\delta_u,
\]
which yields Eq.~\eqref{eq:window_cancel}.
\end{proof}

\begin{proof}[Proof of Theorem~\ref{thm:mixed_fidelity_bound}]
Let
\begin{equation}
    q=\|\rho\|_\infty
\end{equation}
and define the high-entanglement-energy tail
\begin{equation}
    D_L
    \equiv
    \mathbf 1_{(K_{\min}+L,\infty)}(K_\rho).
\end{equation}
Its weight in the original state is
\begin{equation}
    w_L\equiv\Tr(\rho D_L).
\end{equation}
On this subspace,
\begin{equation}
    \rho
    \preceq
    qe^{-L}D_L,
\end{equation}
and therefore, spectrally,
\begin{equation}
    \sqrt{\rho}
    \succeq
    \frac{\rho}{\sqrt{qe^{-L}}}
    \qquad\text{on }D_L.
\end{equation}
It follows that
\begin{equation}
    1-h
    =
    \frac{1}{\sqrt r}\Tr\sqrt\rho
    \geq
    \frac{w_Le^{L/2}}{\sqrt{rq}},
\end{equation}
or
\begin{equation}
    w_L
    \leq
    (1-h)\sqrt{rq}\,e^{-L/2}
    \leq
    A_r(h)e^{-L/2}.
    \label{eq:deep_tail_bound}
\end{equation}

where $A_r(h)\equiv(1-h)\sqrt{r\,q_{\max}(h)}$. For $u\geq1$, tilting by $\rho^u$ can only decrease the weight of the
high-entanglement-energy tail. Indeed,
\begin{equation}
    \delta_u\equiv\Tr(\rho_uD_L)
\end{equation}
satisfies
\begin{equation}\label{eq:u-derivative of D support}
    \frac{d\delta_u}{du}
    = - (\langle K_\rho D_L \rangle_u - \langle K_\rho\rangle_u \langle D_L \rangle_u),
\end{equation}
where $\langle . \rangle_u  $ is expectation computed in the state $\rho_u$. Define the mean energies in the tail and its complement by
\[
    \overline K_{\rm tail}
    =\frac{\langle K_\rho D_L\rangle_u}{\delta_u},
    \qquad
    \overline K_{\rm rest}
    =\frac{\langle K_\rho(1-D_L)\rangle_u}{1-\delta_u}.
\]
Then
\[
    \frac{d\delta_u}{du}
    =-\delta_u(1-\delta_u)
    (\overline K_{\rm tail}-\overline K_{\rm rest})\leq0,
\]
since every energy in the tail exceeds every energy in its
complement. Hence $\delta_u\leq\delta_1=w_L$ for $u\geq1$. Suppose now that $f_\rho(u+iv)=0$. For any
$\theta\in(0,\pi/2)$ choose
\begin{equation}
    L=\frac{2\theta}{|v|}.
\end{equation}
Lemma~\ref{lem:window_cancellation} and
Eq.~\eqref{eq:deep_tail_bound} imply
\begin{equation}
    \frac{\cos\theta}{1+\cos\theta}
    \leq
    A_r(h)e^{-\theta/|v|}.
\end{equation}
Whenever
\begin{equation}
    A_r(h)\frac{1+\cos\theta}{\cos\theta}>1,
\end{equation}
this rearranges to
\begin{equation}
    |v|
    \geq
    \frac{\theta}{
    \log\!\left[
        A_r(h)
        \frac{1+\cos\theta}{\cos\theta}
    \right]}.
\end{equation}
Maximizing over $\theta$ proves
Eq.~\eqref{eq:mixed_master_width}. If for some admissible $\theta$ the
quantity inside the logarithm is at most unity, the necessary inequality
$1\leq A_r(h)(1+\cos\theta)e^{-\theta/|v|}/\cos\theta$
is impossible. Hence no zero with $u\geq1$ exists.

For fixed $h$ and $r\rightarrow\infty$,
\begin{equation}
    q_{\max}(h)
    =
    h(2-h)+O(r^{-1/2}),
\end{equation}
so
\begin{equation}
    A_r(h)
    =
    \sqrt{r\,g(h)}\,[1+o(1)],
    \qquad
    g(h)=h(2-h)(1-h)^2.
\end{equation}
The trivial bound $\theta<\pi/2$ gives
$B_r(h)\leq\pi/[2\log A_r(h)]$ at leading order. Conversely, choosing
$\theta=\pi/2-O(1/\log A_r)$ gives the same leading value, with a
relative correction of order
$\log\log r/\log r$. Since
$2\log A_r=\log[r\,g(h)]+o(1)$, this yields
Eq.~\eqref{eq:mixed_master_asymptotic}.
\end{proof}

Finally,
\begin{equation}
    g'(h)
    =
    -2(h-1)(2h^2-4h+1),
\end{equation}
so its unique interior maximum occurs at
\begin{equation}
    h_\star=1-\frac{1}{\sqrt2},
\end{equation}
with $g(h_\star)=1/4$. For the critical two-sector state defined in Eq.~\eqref{eq:critical two sector state definition},
\begin{align}
    \sqrt{F_{\mathrm{mix}}(\rho_\star)}
    &=
    \frac{1}{\sqrt r}
    \left[
        \frac{1}{\sqrt2}
        +
        (r-1)
        \frac{1}{\sqrt{2(r-1)}}
    \right]
    \nonumber\\
    &=
    \frac{1+\sqrt{r-1}}{\sqrt{2r}}
    \xrightarrow{r\to\infty}
    \frac{1}{\sqrt2},
\end{align}
Thus $h(\rho_\star)\to h_\star=1-1/\sqrt2$, the value
that maximizes $g(h)$. The critical two-sector state
therefore asymptotically realizes this extremal
fidelity parameter.

\subsection{Typical Haar-state clearance}
\label{app:haar_clearance}

\begin{proof}[Proof of Theorem~\ref{thm:haar_clearance}]
Write $N=d_A$ and $\rho_N=\rho_A$. Let $\lambda_{j,N}$ be the
$N$ eigenvalues of $\rho_N$, and define the rescaled eigenvalues and
empirical measure
\begin{equation}
    x_{j,N}=N\lambda_{j,N},
    \qquad
    \mu_N=\frac1N\sum_{j=1}^{N}\delta_{x_{j,N}}.
    \label{eq:app_haar_empirical_measure}
\end{equation}
We use $x^z=e^{z\log x}$ for $x>0$; any zero eigenvalues contribute
zero when $\Re z>0$. The normalized moment function is
\begin{equation}
    M_N(z)
    \equiv N^{z-1}f_{\rho_N}(z)
    =\int x^z\,d\mu_N(x).
    \label{eq:app_haar_normalized_moment}
\end{equation}
The prefactor never vanishes, so $M_N$ and $f_{\rho_N}$ have the same
Fisher zeros. We first exclude zeros from bounded regions of the
continuation half-plane, then use the entropy-gap theorem to control
arbitrarily large $\Re z$.

\emph{Limiting spectrum and complex moments.}
For Haar-random reduced states with $d_A/d_B \to c\in(0,1]$, the
empirical spectral measure $\mu_{d_A}$ of the rescaled state $d_A\rho_A$, as $d_A \to \infty$ at fixed $c$, converges weakly in probability to the
Mar\v{c}enko--Pastur law, and the largest rescaled eigenvalue converges
to its upper edge~\cite{Nechita2007}:
\begin{equation}
    \mu_N\Longrightarrow\mu_c,
    \qquad
    x_{\max,N}\equiv\max_j x_{j,N}
    \xrightarrow{\mathrm{prob.}}x_+,
    \label{eq:app_haar_spectral_convergence}
\end{equation}
where
\begin{align}
    d\mu_c(x)
    &=\frac{\sqrt{(x_+-x)(x-x_-)}}{2\pi c x}
      \mathbf1_{(x_-,x_+)}(x)\,dx,
    \nonumber\\
    x_\pm&=(1\pm\sqrt c)^2.
    \label{eq:app_haar_MP_measure}
\end{align}

Consequently $M_N(z)\to M_c(z)\equiv\int x^z\,d\mu_c(x)$ locally uniformly in
probability on $\{\Re z>0\}$.

\begin{proof}
Fix $B>x_+$; since $x_{\max,N}\to x_+$ in probability,
$\operatorname{supp}\mu_N\subset[0,B]$ with probability $\to1$. For fixed $z$ with
$\Re z>0$ the map $x\mapsto x^z$ is bounded and continuous on $[0,B]$
($x^z\to0$ as $x\to0$), so $\mu_N\Rightarrow\mu_c$ gives
\begin{equation}
    M_N(z)=\int x^z\,d\mu_N \xrightarrow{\mathrm{prob.}} \int x^z\,d\mu_c=M_c(z).
    \label{eq:app_haar_moment_pointwise}
\end{equation}
On $[0,B]$, $|\partial_z x^z|=|x^z\log x|$ is bounded uniformly for $z$ in a
compact $\mathcal K\subset\{\Re z>0\}$, hence $M_N$ and $M_c$ are Lipschitz on
$\mathcal K$ with a common constant independent of $N$:
\begin{equation}
    |M_N(z)-M_N(z')|\le C|z-z'|,\quad |M_c(z)-M_c(z')|\le C|z-z'|.
    \label{eq:app_haar_equicontinuity}
\end{equation}
Equicontinuity \eqref{eq:app_haar_equicontinuity} promotes the pointwise limit
\eqref{eq:app_haar_moment_pointwise} to a uniform one: applying it on a finite
net of $\mathcal K$ controls the error everywhere in between, giving
\begin{equation}
    \sup_{z\in\mathcal K}|M_N(z)-M_c(z)|\xrightarrow{\mathrm{prob.}}0 .
    \label{eq:app_haar_compact_convergence}
\end{equation}
\end{proof}

The derivatives converge locally uniformly as well, by Cauchy's
integral formula. This argument includes $c=1$ and does not require a
positive lower bound on the smallest eigenvalue.

The limiting moment function $M_c(z) $ is zero-free on $\Re z>0$. For
$0<c<1$, Proposition~5.1 of Ref.~\cite{cunden2019moments} places all its
zeros on $\Re z=-1/2$. For $c=1$, a beta integral gives directly
\begin{align}
    M_1(z)
    &=\frac{1}{2\pi}\int_0^4
      x^{z-1/2}(4-x)^{1/2}\,dx
    \nonumber\\
    &=\frac{4^z\Gamma(z+\tfrac12)}
           {\sqrt\pi\,\Gamma(z+2)},
    \label{eq:app_haar_balanced_transform}
\end{align}
which has no zeros in the domain $\text{Re}(z) > -\frac{1}{2}$.

\emph{Control at large real replica index.}
Compact convergence alone does not exclude zeros whose real parts
diverge with $N$. To control this possibility, fix $U\geq1$ and
introduce the tilted state
\begin{equation}
    \sigma_{N,U}=\frac{\rho_N^U}{f_{\rho_N}(U)}.
\end{equation}
Its entropy gap, divided by $U$, is
\begin{align}
    g_N(U)
    &\equiv
    \frac{S_{\mathrm{vN}}(\sigma_{N,U})
          -S_\infty(\sigma_{N,U})}{U}
    \nonumber\\
    &=\log x_{\max,N}-\frac{M_N'(U)}{M_N(U)}.
    \label{eq:app_haar_tilted_gap}
\end{align}
Since
\begin{equation}
    \Tr\sigma_{N,U}^{\,w}
    =\frac{f_{\rho_N}(Uw)}{f_{\rho_N}(U)^w},
\end{equation}
Theorem~\ref{thm:entropy_gap_clearance}, applied to $\sigma_{N,U}$,
implies that every zero of $f_{\rho_N}(z)$ with $\Re z\geq U$
satisfies
\begin{equation}
    |\Im z|\geq\frac{\pi}{2g_N(U)}.
    \label{eq:app_haar_large_real_bound}
\end{equation}
If $g_N(U)=0$, the tilted state, and hence $\rho_N$, is flat on its
support and has no Fisher zeros.

For each fixed $U$, Eqs.~\eqref{eq:app_haar_spectral_convergence}
and~\eqref{eq:app_haar_compact_convergence} give
\begin{equation}
    g_N(U)\xrightarrow{\mathrm{prob.}}g_c(U),
    \label{eq:app_haar_gap_convergence}
\end{equation}
where
\begin{align}
    g_c(U)
    &\equiv\log x_+-\frac{M_c'(U)}{M_c(U)}
    \nonumber\\
    &=\frac{\int x^U\log(x_+/x)\,d\mu_c(x)}
           {\int x^U\,d\mu_c(x)}.
    \label{eq:app_haar_gap_limit}
\end{align}
Moreover, $g_c(U)\to0$ as $U\to\infty$: increasing $U$ concentrates
the tilted measure near the upper spectral edge. An elementary bound
makes this precise, including at the balanced hard edge. Let $\nu_c$
be the distribution of $y=\log(x_+/x)$ under $\mu_c$. For any
$\delta>0$, its measure $p_\delta=\nu_c([0,\delta/2])$ is positive.
For $U\geq1/\delta$, splitting the numerator at $y=\delta$ gives
\begin{align}
    0\leq g_c(U)
    &=\frac{\int y e^{-Uy}\,d\nu_c(y)}
           {\int e^{-Uy}\,d\nu_c(y)}
    \nonumber\\
    &\leq\delta+\frac{\delta}{p_\delta}e^{-U\delta/2}.
    \label{eq:app_haar_gap_decay}
\end{align}
Here the denominator is at least $p_\delta e^{-U\delta/2}$, while
$y e^{-Uy}\leq\delta e^{-U\delta}$ for $y\geq\delta$.
Taking $U\to\infty$ and then $\delta\to0$ proves the claim.

\emph{Combining the two regions.}
Fix $T>0$ and set $H=T+1$. Choose a finite $U>1$, independent of
$N$, such that $g_c(U)<\pi/(4H)$.
Equation~\eqref{eq:app_haar_gap_convergence} then implies
\begin{equation}
    \Pr\!\left[g_N(U)<\frac{\pi}{2H}\right]\longrightarrow1.
\end{equation}
On this event, Eq.~\eqref{eq:app_haar_large_real_bound} excludes
all zeros with $\Re z\geq U$ and $|\Im z|\leq H$.

The remaining rectangle
\begin{equation}
    \mathcal R_{U,H}
    =\{z:\ 1\leq\Re z\leq U,\ |\Im z|\leq H\}
\end{equation}
is compact, and $M_c$ has no zeros there. Hence
$m_{U,H}\equiv\min_{\mathcal R_{U,H}}|M_c|>0$.
By Eq.~\eqref{eq:app_haar_compact_convergence}, with probability
tending to one,
\begin{equation}
    \sup_{\mathcal R_{U,H}}|M_N-M_c|<\frac{m_{U,H}}{2},
\end{equation}
which excludes zeros from this rectangle as well. Together, the two
regions cover $\{\Re z\geq1,\ |\Im z|\leq H\}$. On their joint
high-probability event, $\Delta_{\mathrm F}(\rho_N)\geq H>T$.
Therefore,
\begin{equation}
    \Pr\!\left[\Delta_{\mathrm F}(\rho_N)\leq T\right]
    \longrightarrow0.
\end{equation}
Since $T$ was arbitrary, this proves the theorem.
\end{proof}

\section{Properties of the SAC Estimator}\label{sec:Appendix D}

In this appendix we establish several analytic properties of the SAC estimator
of Eq.~\eqref{eq:SAC_estimator}. We first explain why the boundary norm of
Eq.~\eqref{eq:normB} is the natural object for the von Neumann problem. Among the
norms built from the discrepancy function, it is the one that detects the
boundary singularity at $w=-1$, and this singularity is cancelled precisely when
the continuation parameter equals $S_{\mathrm{vN}}$. We then derive an exact
expression for the SAC error and use it to show that the estimator is exact on
affine R\'enyi data. Finally, we show that the estimator is the minimum-norm
interpolant in a reproducing kernel Hilbert space (RKHS) whose kernel generates
the matrix $A$, so that the norm, the matrix $A$, and the estimator are three
facets of a single Hilbert-space structure.

\subsection{Choice of Norm}\label{sec:choice of norm}

We now discuss why the norm in Eq.~\eqref{eq:normB} is relevant for the von
Neumann entropy problem. The Ciulli--Spearman construction relies on a
Hilbert-space structure, and Eq.~\eqref{eq:normB} defines an \(L^2\) norm
induced by an inner product, which is required for the duality theorem of
Refs.~\cite{CS1,CS2}. Additionally, this norm is useful because it
detects the boundary singularity associated with the point $w=-1 (\text{equivalently,} \hspace{2pt} z =1)$, in the von Neumann entropy problem, as we shall explain below. 

As pointed out in Sec.~\ref{sec:the methodology of SAC}, we choose the function $F(w)$ in Eq.~\eqref{eq:normB} to be the discrepancy function $D_\alpha(z(w))$. This discrepancy function in the unit disc is obtained from
\(D_\alpha(z)\), defined in Eq.~\eqref{eq:Discrepancy function}, through the
conformal map described in Sec.~\ref{sec:justification of choice of parameters}. 
Since \(z=1\) is mapped to \(w=-1\), a pole of \(D_\alpha(z)\) at \(z=1\)
appears as a boundary singularity of \(D_\alpha(z(w))\) at \(w=-1\). Thus, the
singular behavior relevant to Eq.~\eqref{eq:normB} can be analyzed directly in
the original \(z\)-plane.

We want to show that the integrand in Eq.~\eqref{eq:normB} diverges as one
approaches \(w=-1\), namely
\begin{equation}
\label{eq:derivative_of_the_imaginary_part_diverges}
    \frac{d}{d\theta}
    \operatorname{Im}
    \left[
        D_\alpha\!\left(z(e^{i\theta})\right)
    \right]
    \longrightarrow \infty,
    \qquad
    \theta\to\pi .
\end{equation}
This would imply that the norm \(\|D_\alpha\|\) is formally divergent unless the pole at
\(w=-1\) is canceled. This cancellation occurs exactly when
\(\alpha=S_{z\to 1}=S_{\mathrm{vN}}\). This motivates using the norm in Eq.~\eqref{eq:normB} to drive
the reconstruction toward $\alpha=S_{\mathrm{vN}}$.

To see Eq.~\eqref{eq:derivative_of_the_imaginary_part_diverges} explicitly, write \(z=x+i y\). Using the definition of the
discrepancy function,
\begin{equation}
\begin{split}
    D_\alpha(z)
    &=
    \frac{S_z(\rho)}{z-1}
    -
    \frac{\alpha}{z-1}
    \\
    &=
    -\frac{
        \log\!\left[
            \operatorname{Tr}\!\left(\rho^{x+i y}\right)
        \right]
    }{
        (x-1+i y)^2
    }
    -
    \frac{\alpha}{x-1+i y}.
\end{split}
\label{eq:Dalpha_xy}
\end{equation}
Approaching \(w=e^{i\theta}=-1\) along the unit circle corresponds, in the
\(z\)-plane, to approaching \(z=1\) with \(y\to 0\). Up to a finite, nonzero
Jacobian factor from the conformal map, the angular derivative in the
\(w\)-plane therefore has the same singular behavior as the derivative with
respect to \(y\) in the \(z\)-plane.

Expanding Eq.~\eqref{eq:Dalpha_xy} about \(z=1\), we find that, for generic
\(\alpha\),
\begin{equation}
\label{eq:Dalpha_pole_expansion}
    D_\alpha(1+i y)
    =
    \frac{C_\alpha}{i y}
    +
    O(1),
    \qquad
    y\to 0,
\end{equation}
where \(C_\alpha\) is finite, real-valued, and independent of \(y\). The leading
singular term is therefore purely imaginary:
\[
    \operatorname{Im}D_\alpha(1+i y)
    \sim
    -\frac{C_\alpha}{y},
    \qquad
    \operatorname{Re}D_\alpha(1+i y)=O(1).
\]
This already shows that a norm based directly on
\(\left|\operatorname{Im}D_\alpha\right|\) would also be sensitive to the pole.
The norm in Eq.~\eqref{eq:normB} probes the derivative of the same singular
quantity. Indeed, differentiating Eq.~\eqref{eq:Dalpha_pole_expansion} gives
\begin{equation}
\label{eq:ddy}
    \left.
    \frac{\partial D_\alpha(z)}{\partial y}
    \right|_{x=1,\;y\to 0}
    =
    \frac{i C_\alpha}{y^2}
    +
    O(1),
\end{equation}
so that \(\partial_y \operatorname{Im}D_\alpha(z)\propto 1/y^2\) unless
\(C_\alpha=0\). Equivalently, after mapping to the unit disc,
Eq.~\eqref{eq:derivative_of_the_imaginary_part_diverges} follows. Thus the norm in Eq.~\eqref{eq:normB}, evaluated on $D_\alpha$, is sensitive to
the boundary singularity at $w=-1$ and is driven toward the pole-cancelling value
$\alpha=S_{\mathrm{vN}}$.

It is also useful to contrast this choice with norms constructed from the real
part of the discrepancy function. For instance, one could consider\footnote{The original papers on SAC considered only the norms constructed from the imaginary part of the function or its angular derivative. However, it is easy to show that working with the norms defined in Eqs.~\eqref{eq:normA} and \eqref{eq:normC} lead to a very similar reconstruction as outlined in~\cite{CS1, CS2}.}
\begin{equation}
\label{eq:normA}
    \frac{1}{2\pi}
    \int_0^{2\pi}
    \left|
        \operatorname{Re}D_\alpha\!\left(z(e^{i\theta})\right)
    \right|^2
    \sigma(\theta)\,d\theta ,
\end{equation}
or
\begin{equation}
\label{eq:normC}
    \frac{1}{2\pi}
    \int_0^{2\pi}
    \left|
        \frac{d}{d\theta}
        \operatorname{Re}D_\alpha\!\left(z(e^{i\theta})\right)
    \right|^2
    \sigma(\theta)\,d\theta .
\end{equation}
However, Eq.~\eqref{eq:Dalpha_pole_expansion} shows that the leading pole is
purely imaginary. Consequently, both $\operatorname{Re}D_\alpha$ and
$\frac{d}{d\theta}\operatorname{Re}D_\alpha$ remain finite as $\theta\to\pi$, so
the norms in Eqs.~\eqref{eq:normA} and~\eqref{eq:normC} are insensitive to the
physical pole. This sensitivity argument concerns the physical discrepancy
function $D_\alpha$, which for generic $\alpha$ is singular at $w=-1$. For the
finite-data interpolants actually minimized over, the real- and
imaginary-derivative norms in fact coincide (Appendix~\ref{sec:RKHS_SAC}), so
this calculation motivates the imaginary-part norm through its sensitivity to the
pole rather than implying that other derivative norms give different estimators. With this motivation, we adopt Eq.~\eqref{eq:normB} throughout this work
and find that the resulting estimator performs robustly in all of our numerical
tests. We leave the systematic study of a norm based directly on
$\left|\operatorname{Im}D_\alpha\right|$ for future work.

\subsection{The SAC error and exactness on affine data}
\label{subsec:affine-data}

It is convenient to introduce
\begin{equation}
\begin{split}
  g_i \equiv \frac{S_i}{z_i-1},\qquad
  h_i \equiv \frac{1}{z_i-1},\qquad  \\
  \langle x,y\rangle \equiv \sum_{i,j}(A^{-1})_{ij}\,(x_i-x_*)(y_j-y_*),
  \label{eq:innprod}
\end{split}
\end{equation}
where $x_*$ denotes the value of $x$ at the subtraction point
$w_0=w_1$ (equivalently $z_1$), the first data point. Because $A$ is symmetric
positive definite, $\langle\cdot,\cdot\rangle$ is a genuine inner product on data vectors modulo constants, equivalently on their differences
relative to the subtraction point, and
the closed-form estimator Eq.~\eqref{eq:SAC_estimator} becomes
\begin{equation}
  S^{\rm SAC}_{vN}=\frac{\langle g,h\rangle}{\langle h,h\rangle}.
  \label{eq:sac-regression}
\end{equation}
This form yields an exact expression for the estimation error. Define the
regular discrepancy function $D^\ast(z)\equiv\dfrac{S_z-S_{vN}}{z-1}$, i.e.
$D_\alpha(z)$ evaluated at $\alpha=S_{vN}$; by construction
$g_i=S_{vN}\,h_i+D^\ast_i$. Substituting into Eq.~\eqref{eq:sac-regression} and
using bilinearity,
\begin{equation}
  S^{\rm SAC}_{vN}-S_{vN}=\frac{\langle D^\ast,\,h\rangle}{\langle h,\,h\rangle},
  \label{eq:exact-error}
\end{equation}
so the error is the SAC-metric projection of the regular part of the data onto
the pole function $h$.

Equation~\eqref{eq:exact-error} immediately gives a noiseless exactness result.
If the R\'enyi function is affine, $S_z=a+bz$, then
$D^\ast(z)=\frac{a+bz-(a+b)}{z-1}=b$ is constant, so $D^\ast_i-D^\ast_*=0$ and
\begin{equation}
  S^{\rm SAC}_{vN}=S_{vN}=a+b \qquad\text{for every }N\ge2 .
  \label{eq:affine-exact}
\end{equation}
Thus the estimator reproduces \emph{any} affine-in-$z$ R\'enyi sequence exactly,
independent of $A$ or the number of data points --- the noiseless analogue of
correctly extrapolating a straight line.

\subsection{The Reproducing Kernel Hilbert space (RKHS) approach to SAC}
\label{sec:RKHS_SAC}

In this section we show that the SAC estimator admits a natural interpretation
in terms of reproducing kernel Hilbert spaces (RKHS)~\cite{aronszajn1950theory}. This
unifies the three ingredients of the construction: the stabilizing norm of
Eq.~\eqref{eq:normB} is the Hilbert-space norm of the kernel underlying the
matrix \(A\), the matrix \(A\) itself is the Gram matrix of that kernel at the
interpolation points, and by the representer
theorem~\cite{scholkopf2001generalized}, for each candidate
value of the continuation parameter \(\alpha\) the SAC minimization selects the
minimum-norm analytic interpolant compatible with the R\'enyi data. The
stabilization criterion is therefore not an ad hoc regularization but the
geometry of minimum-norm interpolation in this space.

\subsubsection{The SAC matrix as an RKHS Gram matrix}

We begin with the closed-form expression for the matrix \(A\). For real 
\(a\in(-1,1)\),
\begin{equation}
    \ln |e^{i\theta}-a|
    =
    -\sum_{k\geq 1}\frac{a^k}{k}\cos(k\theta).
\end{equation}
Using the orthogonality of the Fourier modes gives
\begin{equation}
\label{eq:log_kernel_identity}
    \frac{1}{2\pi}
    \int_{0}^{2\pi}
    \ln|e^{i\theta}-a|
    \ln|e^{i\theta}-b|
    \,d\theta
    =
    \frac{1}{2}\operatorname{Li}_2(ab),
\end{equation}
where
\begin{equation}
    \operatorname{Li}_s(x)
    =
    \sum_{k\geq1}\frac{x^k}{k^s}
\end{equation}
is the polylogarithm. It follows immediately from 
Eq.~\eqref{eq:A_matrix_elements} that
\begin{equation}
\label{eq:A matrix in terms of dilogarithms}
    A_{ij}
    =
    2\left[
    \operatorname{Li}_2(w_iw_j)
    -
    \operatorname{Li}_2(w_iw_1)
    -
    \operatorname{Li}_2(w_jw_1)
    +
    \operatorname{Li}_2(w_1^2)
    \right].
\end{equation}

This expression naturally defines an RKHS. Since the SAC construction only 
depends on differences of function values, the constant component of the 
function plays no role. We therefore fix this irrelevant degree of freedom and 
consider the space
\begin{equation}
\label{eq:RKHS_space}
    \mathcal H
    =
    \left\{
    f(w)=\sum_{k\geq1}a_k w^k
    \;:\;
    \frac{1}{2}\sum_{k\geq1}k^2|a_k|^2<\infty
    \right\},
\end{equation}
equipped with the inner product
\begin{equation}
\label{eq:RKHS_inner_product}
    \langle f,g\rangle_{\mathcal H}
    =
    \frac{1}{2}
    \sum_{k\geq1}k^2a_k\overline{b_k},
    \qquad
    g(w)=\sum_{k\geq1}b_k w^k.
\end{equation}
The corresponding norm is
\begin{equation}
    \|f\|_{\mathcal H}^2
    =
    \frac{1}{2}\sum_{k\geq1}k^2|a_k|^2.
\end{equation}

The reproducing kernel of this space is
\begin{equation}
\label{eq:RKHS_kernel}
    K(w,b)
    =
    2\operatorname{Li}_2(w\overline b)
    =
    2\sum_{k\geq1}\frac{\overline b^{\,k}}{k^2}w^k,
\end{equation}
since
\begin{equation}
    \langle f,K(\cdot,b)\rangle_{\mathcal H}
    =
    \sum_{k\geq1}a_kb^k
    =
    f(b).
\end{equation}

The quantities entering SAC are not the values \(f(w_i)\) individually, but 
their differences relative to the first interpolation point \(w_1\). Define 
the linear functionals
\begin{equation}
    L_i[f]
    \equiv
    f(w_i)-f(w_1),
    \qquad i\geq2.
\end{equation}
By the reproducing property, the representer of \(L_i\) is
\begin{equation}
\label{eq:subtracted_kernel}
    r_i(w)
    =
    K(w,w_i)-K(w,w_1),
\end{equation}
so that
\begin{equation}
    L_i[f]
    =
    \langle f,r_i\rangle_{\mathcal H}.
\end{equation}

Since the SAC interpolation points \(w_i\) are real, the Gram matrix of these 
representers is
\begin{equation}
\begin{split}
\label{eq:A as a gram matrix}
    \langle r_i,r_j\rangle_{\mathcal H}
    &=
    2\sum_{k\geq1}
    \frac{(w_i^k-w_1^k)(w_j^k-w_1^k)}{k^2}
    \\
    &=
    2\bigg[
    \operatorname{Li}_2(w_iw_j)
    -
    \operatorname{Li}_2(w_iw_1) \\
    &\quad -
    \operatorname{Li}_2(w_jw_1)
    +
    \operatorname{Li}_2(w_1^2)
    \bigg]
    \\
    &=A_{ij}.
\end{split}
\end{equation}
Thus, the matrix \(A\) appearing in SAC is precisely the Gram matrix generated 
by the subtracted kernel functions \(\{r_i\}\). In particular, its positivity 
is a direct consequence of the underlying Hilbert-space structure.

\subsubsection{The SAC norm is the RKHS norm}

We next show that the norm defining the RKHS is exactly the boundary norm used 
in SAC, as defined in Eq.~\eqref{eq:normB}. By Parseval's identity,
\begin{equation}
\label{eq:RKHS_boundary_norm}
\begin{split}
    \|f\|_{\mathcal H}^2
    &=
    \frac{1}{2}\sum_{k\geq1}k^2|a_k|^2
    \\
    &=
    \frac{1}{4\pi}
    \int_0^{2\pi}
    |f'(e^{i\theta})|^2\,d\theta.
\end{split}
\end{equation}

On the other hand, the SAC norm in Eq.~\eqref{eq:normB} can be written as
\begin{equation}
\label{eq:SAC_boundary_norm_RKHS}
    \|f\|_{\mathrm{SAC}}^2
    =
    \frac{1}{2\pi}
    \int_0^{2\pi}
    \left|
    \frac{d}{d\theta}
    \operatorname{Im}f(e^{i\theta})
    \right|^2
    \,d\theta.
\end{equation}
Since
\begin{equation}
    \frac{d}{d\theta}
    \operatorname{Im}f(e^{i\theta})
    =
    \operatorname{Re}
    \left[
    e^{i\theta}f'(e^{i\theta})
    \right],
\end{equation}
we obtain
\begin{equation}
\begin{split}
    \|f\|_{\mathrm{SAC}}^2
    &=
    \frac{1}{2\pi}
    \int_0^{2\pi}
    \left[
    \operatorname{Re}
    \left(
    e^{i\theta}f'(e^{i\theta})
    \right)
    \right]^2
    \, d\theta.
\end{split}
\end{equation}
We can rewrite the integrand as
\begin{equation}
\begin{split}
    \left[
    \operatorname{Re}
    \left(
    e^{i\theta}f'(e^{i\theta})
    \right)
    \right]^2
    &=
    \frac{1}{4}
    \Bigg[
    \left(
    e^{i\theta}f'(e^{i\theta})
    \right)^2
    \\
    &\qquad
    +
    \left(
    e^{-i\theta}\overline{f'(e^{i\theta})}
    \right)^2
    +
    2|f'(e^{i\theta})|^2
    \Bigg].
\end{split}
\end{equation}

Since \((w f'(w))^2\) is analytic inside the unit disc and
vanishes at \(w=0\), the first term has vanishing angular average by the mean-value property for
holomorphic functions. 
The second term is the complex conjugate of the first and therefore also has
vanishing angular average. Hence, only the absolute-value term survives, giving
\begin{equation}
    \|f\|_{\mathrm{SAC}}^2
    =
    \frac{1}{4\pi}
    \int_0^{2\pi}
    |f'(e^{i\theta})|^2
    \, d\theta
    =
    \|f\|_{\mathcal H}^2.
\end{equation}
Thus, the SAC norm in Eq.~\eqref{eq:normB} is exactly the canonical RKHS norm associated 
with the dilogarithm kernel in Eq.~\eqref{eq:RKHS_kernel}. The choice of 
stabilization criterion and the matrix \(A\) are consequently two aspects of 
the same underlying Hilbert-space geometry. 
The same computation applies to the real part. Since
\begin{equation}
    \frac{d}{d\theta}\operatorname{Re}f(e^{i\theta})
    =-\operatorname{Im}\!\left[e^{i\theta}f'(e^{i\theta})\right],
\end{equation}
the identical mean-value cancellation gives
\begin{equation}
    \frac{1}{2\pi}\int_0^{2\pi}
    \left|\frac{d}{d\theta}\operatorname{Re}f(e^{i\theta})\right|^2 d\theta
    =\frac{1}{4\pi}\int_0^{2\pi}|f'(e^{i\theta})|^2\,d\theta
    =\|f\|_{\mathcal H}^2 .
\end{equation}
The real and imaginary part derivative norms (Eqs.~\eqref{eq:normC} and~\eqref{eq:normB} respectively) therefore coincide on
$\mathcal H$. This clarifies the role of the pole-detection argument in
Sec.~\ref{sec:choice of norm}. On the finite-norm space $\mathcal H$, over which
the SAC minimization is actually performed, the two derivative norms define the
same estimator. The equivalence relies on $f$ having finite norm, but the
discrepancy function $D_\alpha$ does not belong to $\mathcal H$, since it is
singular at $w=-1$ for generic $\alpha$. As we highlighted in Appendix~\ref{sec:choice of norm}, the derivative of imaginary-part norm is chosen for $D_\alpha$
for its manifest sensitivity to that pole. For the discrepancy function $D_\alpha$, cancellation
of that pole would fix $\alpha=S_{\mathrm{vN}}$.
SAC instead minimizes the norm of interpolants
that match finitely many discrepancy values,
so its estimate need not coincide with the
exact pole-canceling value, consistent with the nonzero error
in Eq.~\eqref{eq:exact-error}.

\subsubsection{SAC as minimum-norm analytic interpolation}

We can now use the Representer theorem~\cite{scholkopf2001generalized} of RKHS to give the SAC estimator a variational 
interpretation. For a fixed value of \(\alpha\), the R\'enyi data impose the 
interpolation constraints
\begin{equation}
\label{eq:RKHS_constraints}
    L_i[f]
    =
    f(w_i)-f(w_1)
    =
    a_i(\alpha) - a_1(\alpha),
    \quad i=2,\ldots,N,
\end{equation}
where \(a_i(\alpha)\) are the data values defined in 
Sec.~\ref{sec:the methodology of SAC} (Note, for the R\'enyi analytic continuation problem, $a_1(\alpha)$ is associated to $S_2(\rho)$, the second R\'enyi entropy). Among all functions in \(\mathcal H\) 
satisfying these constraints, consider the minimum-norm interpolation 
problem
\begin{equation}
\label{eq:RKHS_min_problem}
    E_{\min}(\alpha)
    =
    \operatorname*{min}_{\substack{
        f\in\mathcal H\\
        L_i[f]=a_i(\alpha),\;i=2,\ldots,N
    }}
    \|f\|_{\mathcal H}^2.
\end{equation}

The Representer theorem states that the minimizer lies entirely in the 
finite-dimensional span of the representers of the interpolation functionals. 
Consequently,
\begin{equation}
\label{eq:representer_solution}
    f_\alpha^\star(w)
    =
    \sum_{j=2}^{N}c_jr_j(w).
\end{equation}
Define the difference vector $d_i(\alpha)\equiv a_i(\alpha)-a_1(\alpha)$. Applying
the constraints gives
\begin{equation}
    d_i(\alpha)=L_i[f^\star_\alpha]=\langle f^\star_\alpha,r_i\rangle_{\mathcal H}
    =\sum_{j=2}^{N}A_{ij}c_j,
\end{equation}
so that $\mathbf c=A^{-1}\mathbf d(\alpha)$ with
$\mathbf d(\alpha)=(a_2(\alpha)-a_1(\alpha),\ldots,a_N(\alpha)-a_1(\alpha))^T$. The
norm of the minimum-norm interpolant is then
\begin{equation}
    E_{\min}(\alpha)=\|f^\star_\alpha\|_{\mathcal H}^2=\mathbf c^T A\,\mathbf c
    =\mathbf d(\alpha)^T A^{-1}\mathbf d(\alpha)=\delta_{\min}^2(\alpha),
\end{equation}
Comparing with Eq.~\eqref{eq:delta2min expression}, we 
thus see that $E_\text{min}(\alpha) = \delta^2_\text{min} (\alpha)$. Minimizing over the remaining one-parameter family labeled by \(\alpha\) consequently 
reproduces the SAC estimator in Eq.~\eqref{eq:SAC_estimator}.

This provides a complementary interpretation of 
stabilized analytic continuation. For every candidate value of the unknown 
continuation parameter \(\alpha\), the finite R\'enyi data define an 
interpolation problem with, in general, infinitely many analytic solutions. 
The RKHS structure assigns a notion of complexity to these solutions through 
the SAC norm, and the Representer theorem identifies the unique minimum-norm 
interpolant. SAC then selects the value of \(\alpha\) for which the data admit 
the least-complex interpolation in this sense. This interpretation concerns the minimum-norm interpolant constructed from the finite set of available data and does not require the underlying physical continuation itself to belong to the RKHS.

\section{Numerical Choices in SAC Implementation} 
\label{sec:justification of choice of parameters}

In this Appendix, we explain the numerical choices underlying our SAC implementation and their effects on accuracy and stability of the estimate. We first motivate the use of a strip geometry from a numerical point of view, and discuss the choice of conformal-map parameters $(\epsilon,\eta)$. We then examine how the clustering of mapped data points limits the number of R\'enyi inputs that can be used reliably. Finally, we compare different singularity exponents in the discrepancy function to motivate our use of a simple pole.

\subsection{Why the numerical implementation uses a strip geometry}
\label{app:strip-versus-wedge}

Section~\ref{sec:domain_analyticity} gives a universal wedge-shaped lower bound on the
zero-free region of $\operatorname{Tr}(\rho^z)$. In practice, any domain within the wedge can be chosen as the analytic domain, and different admissible geometries produce
different images $w_i=w(z_i)$ of the same R\'enyi indices and therefore
different matrices $A$. Since the minimized norm is defined through
$A^{-1}$, the choice of domain shape affects the sensitivity to the
variational parameter $\alpha$ even when the input data in the
$z$-plane are unchanged. We now provide an argument why we chose the strip geometry instead of the original wedge geometry.

As we saw in appendix~\ref{sec:Appendix D}, we can write the SAC estimator in the form given by eq.~\eqref{eq:sac-regression} prompting us to 
rewrite~\eqref{eq:delta2min expression} as
\begin{equation}
    \delta^2(\alpha)
    =\langle g-\alpha h,g-\alpha h\rangle.
    \label{eq:geometry-alpha-norm}
\end{equation}
It is easy to check that
\be
\delta^2(\alpha) - \delta^2(\alpha_{\min}) = \langle h, h \rangle \left( \alpha - \alpha_{\min}\right)^2
\ee
Thus, $\langle h,h\rangle$ measures the sensitivity of
the minimized norm to $\alpha$: a larger $\langle h,h\rangle$ produces a narrower,
more sharply defined minimum. The sequence $h_i=1/(z_i-1)$ itself is
independent of the choice of analytic domain or the conformal map. All domain-geometry dependence enters through
$A^{-1}$ in Eq.~\eqref{eq:innprod}.

In our finite-data continuation problem, 
the strip geometry turns out to produce a larger $\langle h, h\rangle$ , i.e. it penalizes the residual pole more strongly than the wedge geometry. It subsequently gives more accurate
estimates of $S_{\mathrm{vN}}$ in all of our benchmark calculations. We
therefore use the strip map throughout the numerical analysis. 
We now discuss the optimization of conformal mapping parameters $(\epsilon, \eta)$ within the strip
family.

\subsection{Parameters in the conformal mapping}
\label{subsec:choice of conformal parameters}

It is useful to distinguish between the discrepancy function
written in the original \(z\)-plane and the corresponding
function written in the unit disc. We define
\begin{equation}
    D_\alpha(z)
    =
    \frac{S_z-\alpha}{z-1},
    \qquad
    \tilde{D}_\alpha(w)
    =
    \frac{S_{z(w)}-\alpha}{w+1}.
\end{equation}
These two functions are related by
\begin{equation}
    D_\alpha(z(w))
    =
    \tilde{D}_\alpha(w)\,f(w),
    \qquad
    f(w)
    =
    \frac{w+1}{z(w)-1}.
\end{equation}
Thus, \(f(w)\) relates the discrepancy functions associated
with the pole factors \(1/(z-1)\) and \(1/(w+1)\).

Using the two-step conformal map defined in
Eqs.~\eqref{eq:Schwarz-Christoffel}
and~\eqref{eq:second step of conformal map},
\begin{equation}
    \xi(z)
    =
    i\sinh\!\left[\frac{\pi(z-1)}{2\epsilon}\right],
    \qquad
    w(z)
    =
    \frac{\xi(z)-i\eta}{\xi(z)+i\eta},
\end{equation}
where \(\epsilon\) is the strip half-width, the inverse is
\begin{equation}
    \xi(w)
    =
    i\eta\frac{1+w}{1-w},
    \qquad
    z(w)
    =
    1+\frac{2\epsilon}{\pi}
    \operatorname{arcsinh}\!\left[
        -i\xi(w)
    \right].
\end{equation}
Here \(\operatorname{arcsinh}\) denotes the principal branch.
Substituting for \(\xi(w)\) gives
\begin{equation}
    z(w)-1
    =
    \frac{2\epsilon}{\pi}
    \operatorname{arcsinh}\!\left[
        \eta\frac{1+w}{1-w}
    \right].
\end{equation}
Hence,
\begin{equation}
\label{eq:f_w_closed_form}
    f(w)
    =
    \frac{\pi}{2\epsilon}\,
    \frac{1+w}{
        \operatorname{arcsinh}\!\left[
            \eta\frac{1+w}{1-w}
        \right]
    }.
\end{equation}
The apparent singularity at \(w=-1\) is removable, with
\[
    \lim_{w\to-1}f(w)=\frac{\pi}{\epsilon\eta}.
\]

We next examine the behavior of \(f(w)\) numerically.
We fix the strip half-width to \(\epsilon=4\) and compare
two representative choices of the conformal-map parameter,
\(\eta=10^{-2}\) and \(\eta=10^2\).
Figure~\ref{fig:comparison of f(w) for different alpha}
shows the corresponding contour plots of \(|f(w)|\)
in the unit disc. For \(\eta=10^{-2}\), the region near
\(w=-1\) is weighted more strongly than the region near
\(w=1\). For \(\eta=10^2\), the larger values instead
occur farther toward the \(w=1\) side of the disc,
although \(f(w)\to0\) as \(w\to1\).

We repeated this comparison over several orders of
magnitude in \(\eta\). We find that \(\eta\sim10^{-2}\)
gives the desired behavior: it enhances the region near
\(w=-1\), the image of the target point \(z=1\).
However, taking \(\eta\) too small leads to numerical
instabilities. In particular, for \(\eta\lesssim10^{-3}\),
the \(A\) matrix becomes increasingly ill-conditioned.
Decreasing \(\eta\) moves the mapped R\'enyi data points
closer to \(w=1\), increasing their clustering and making
the matrix inversion unstable. Balancing these
considerations, we use \(\eta=10^{-2}\) and
\(\epsilon=4\) for all numerical results presented
in this work.

\begin{figure}
    \centering
    \includegraphics[width=1\linewidth]
    {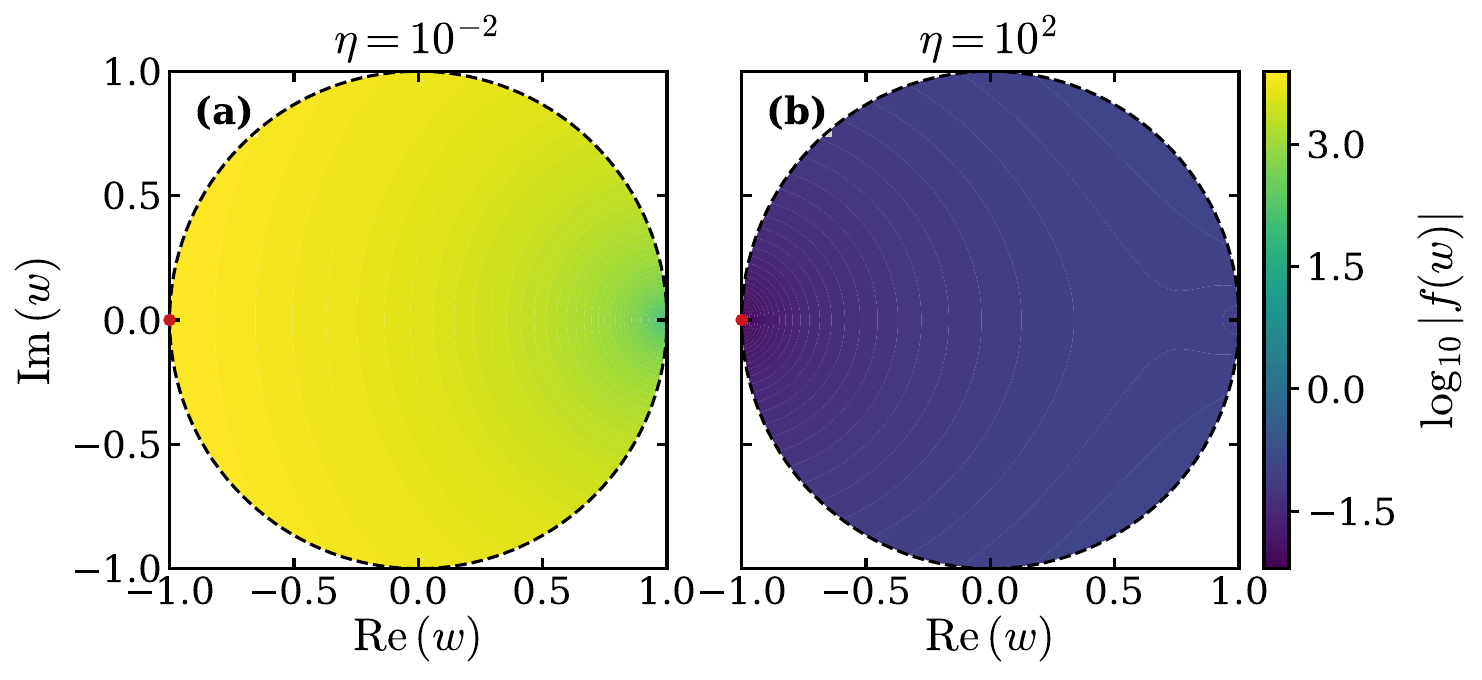}
    \caption{Two-dimensional contour plots of \(|f(w)|\)
    for strip half-width \(\epsilon=4\).
    Panels (a) and (b) correspond to \(\eta=10^{-2}\)
    and \(\eta=10^2\), respectively.}
    \label{fig:comparison of f(w) for different alpha}
\end{figure}

We also examined the effect of varying \(\epsilon\)
while keeping \(\eta=10^{-2}\) fixed.
In Eq.~\eqref{eq:f_w_closed_form}, \(\epsilon\) appears
only through the overall factor \(\pi/(2\epsilon)\).
Therefore, for any two fixed points \(w_1,w_2\) in the
unit disc, the ratio
\[
    \left|\frac{f(w_1)}{f(w_2)}\right|
\]
is independent of \(\epsilon\).
Changing the strip half-width therefore leaves the
relative weighting at fixed points in the \(w\)-plane
unchanged, but affects the positions of the mapped
R\'enyi data. In particular, taking \(\epsilon\) too
small moves these points close to \(w=1\), worsening
the conditioning of \(A\). The order-one choice
\(\epsilon=4\) gives stable performance in our tests, and is justified by our analytics of Sec.~\ref{sec:domain_analyticity}.

Increasing \(\eta\) can partly compensate for the
clustering caused by a small \(\epsilon\).
However, in our numerical implementation, such
parameter choices produce estimated quadrature errors
in entries of \(A\) that exceed the magnitudes of the
entries themselves. The numerical integration used to
construct \(A\) is then unreliable. We avoid these
issues with the parameter choice
\(\epsilon=4\), \(\eta=10^{-2}\).

These parameters give stable reconstructions when
using the discrepancy function \(D_\alpha(z)\).
One could instead formulate the continuation using
\(\tilde{D}_\alpha(w)\) directly. Since the two
discrepancy functions differ by the nonconstant factor
\(f(w)\), the parameter choice should be reassessed
for that formulation. We leave this investigation
for future work.

\subsection{Optimal number of R\'enyi entropies for a given conformal parameter set}
\label{app:Nstar-params}

The number of R\'enyi inputs that can be used reliably
in SAC depends on their locations after the conformal
map. As the integer indices increase, their images
$w_i=w(z_i)$ cluster near the boundary point $w=1$,
which can make the matrix $A$ increasingly ill-conditioned.
We examine how the strip half-width $\epsilon$ and the
M\"obius parameter $\eta$ control this clustering, and
use a numerical resolution criterion to estimate the
practical threshold $N^\ast$ of the number of R\'enyis used in the continuation.

For real $z>1$, the map is real and can be written as
\begin{equation}
    w(z)=\frac{s(z)-\eta}{s(z)+\eta},
    \qquad
    s(z)\equiv
    \sinh\!\left[\frac{\pi(z-1)}{2\epsilon}\right].
    \label{eq:wreal}
\end{equation}
It sends $z=1$ to $w=-1$ and $z\to\infty$ to $w=1$,
with $w(z)=0$ when $s(z)=\eta$.

\begin{figure}
    \centering
    \includegraphics[width=1\linewidth]
    {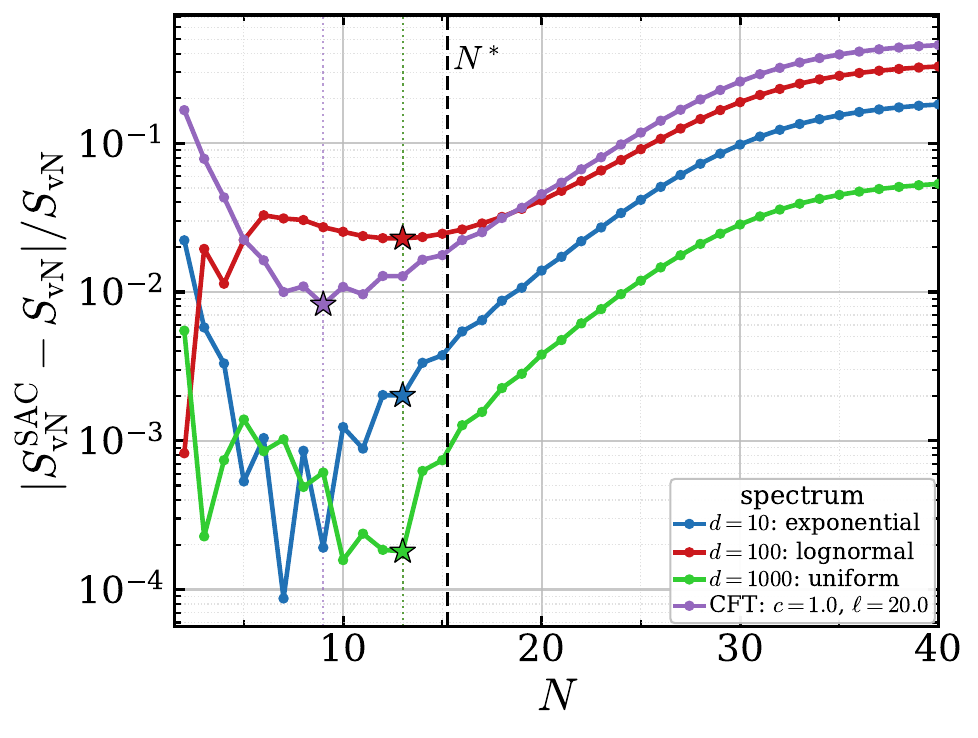}
    \caption{Relative error as a function of the number
    of R\'enyi entropies used in SAC. Colored stars mark
    the optimal number of input R\'enyi entropies for three random spectra of
    dimensions $10$, $100$, and $1000$ chosen from
    different distributions, and for a single-interval
    CFT spectrum. The results are compared with the
    estimated threshold $N^\ast\approx15$ for
    $\eta=10^{-2}$, $\epsilon=4$, and $\delta=10^{-4}$.}
    \label{fig:optimal Nstar}
\end{figure}

\paragraph{Clustering rate and the role of each parameter.}
From Eq.~\eqref{eq:wreal},
\begin{equation}
\begin{split}
    1-w(z)
    &=\frac{2\eta}{s(z)+\eta}
      \sim4\eta\,\varrho^{\,z-1},
      \qquad z\to\infty,\\
    \varrho
    &\equiv e^{-\pi/(2\epsilon)}\in(0,1).
\end{split}
\label{eq:onemw}
\end{equation}
The nearest-neighbor spacing therefore satisfies
\begin{equation}
\begin{split}
    w(z+1)-w(z)
    &=[1-w(z)]-[1-w(z+1)]\\
    &\sim4\eta(1-\varrho)\varrho^{\,z-1},
    \qquad z\to\infty.
\end{split}
\label{eq:spacing}
\end{equation}
Both the distance to $w=1$ and the local spacing decay
geometrically with ratio $\varrho=e^{-\pi/(2\epsilon)}$.
Increasing $\epsilon$ brings $\varrho$ closer to unity
and slows this clustering. At fixed $\epsilon$, $\eta$
controls the prefactor and shifts the mapped data
points, including the first point
\[
    w(2)=\frac{s(2)-\eta}{s(2)+\eta}.
\]
For the parameters used in this work,
$\epsilon=4$ and $\eta=10^{-2}$,
\[
    s(2)=\sinh(\pi/8)\approx0.40\gg\eta,
    \qquad
    w(2)\approx0.95.
\]

\begin{figure}[t!]
    \centering
    \includegraphics[width=1\linewidth]
    {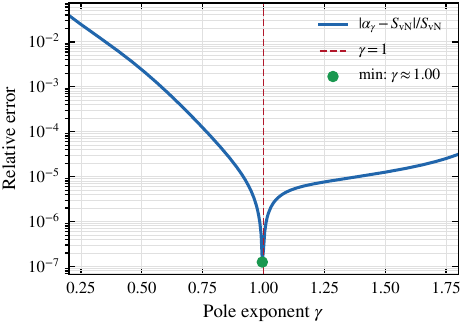}
    \caption{Relative error
    $\lvert\alpha_{\gamma}-S_{\mathrm{vN}}\rvert/S_{\mathrm{vN}}$
    in the analytically continued von Neumann entropy
    as a function of the pole exponent $\gamma$.
    For each $\gamma$, we minimize the boundary norm
    of the generalized discrepancy function
    $D_{\alpha,\gamma}(z)=(S_z-\alpha)/(z-1)^{\gamma}$
    with respect to $\alpha$, where the minimizing
    value $\alpha_{\gamma}$ estimates $S_{\mathrm{vN}}$.
    The error exhibits a sharp minimum at $\gamma\approx1$,
    supporting the simple-pole prescription used in SAC.}
    \label{fig:pole_order_CFT_1interval}
\end{figure}

\paragraph{Estimated input threshold.}
Let $\delta$ denote the tolerance level: the value of $1-w$ (equivalently, the
spacing~\eqref{eq:spacing}) below which an added R\'enyi index, though it still
carries independent information in the $z$-plane, is compressed by the conformal
map so close to $w=+1$ that it becomes numerically indistinguishable from the
indices already present. Once the spacing falls below $\delta$, the mapped data
points pile up near $w=+1$, the Gram matrix $A$ becomes ill-conditioned, and
adding further indices degrades rather than improves the SAC reconstruction of
the boundary value at $w=-1$. We emphasize that $\delta$ is a practical resolution threshold, set by the conditioning of $A$ under the conformal map. It
does not mean the additional inputs cease to carry independent information: that
information remains present in the $z$-plane, but is squeezed below numerical
resolution once the mapped points cluster near $w=+1$.

Using Eq.~\eqref{eq:onemw}, the condition
$1-w(z)\geq\delta$ is equivalent to
\begin{equation}
    z-1
    \leq
    \frac{2\epsilon}{\pi}
    \operatorname{arcsinh}\!\left(
        \frac{2\eta-\eta\delta}{\delta}
    \right).
\end{equation}
For $\delta\ll1$, this gives the estimate
\begin{equation}
    N^\ast
    \approx
    \frac{2\epsilon}{\pi}
    \operatorname{arcsinh}\!\left(
        \frac{2\eta}{\delta}
    \right).
    \label{eq:Nstar-law}
\end{equation}
Here we count consecutive inputs $z=2,\ldots,N+1$,
so the largest value of $z-1$ is $N$.

At fixed $\eta$ and $\delta$, this estimate grows
linearly with $\epsilon$. When $2\eta/\delta\gg1$,
\[
    N^\ast
    \approx
    \frac{2\epsilon}{\pi}
    \log\!\left(\frac{4\eta}{\delta}\right),
\]
so its dependence on $\eta$ is only logarithmic.
The estimate depends on the mapped sampling geometry
rather than the spectrum, although the observed
error minimum can also depend on the spectrum and
the numerical implementation.

The nearest-neighbor spacing provides a related
resolution criterion, but its tolerance differs:
Eq.~\eqref{eq:spacing} gives
\[
    w(z+1)-w(z)
    \sim(1-\varrho)[1-w(z)].
\]
We use the distance criterion above throughout
this estimate.

For $\delta=10^{-4}$, $\epsilon=4$, and $\eta=10^{-2}$,
Eq.~\eqref{eq:Nstar-law} gives
\[
    N^\ast
    \approx
    \frac{8}{\pi}\operatorname{arcsinh}(200)
    \approx15.3.
\]
Figure~\ref{fig:optimal Nstar} compares this estimate
with the observed optimal input counts for random
spectra of dimensions $10$, $100$, and $1000$ drawn
from different distributions, and for a
single-interval CFT spectrum.

The estimate above ties the number of usable R\'enyi
inputs to the chosen strip half-width.
As discussed in Sec.~\ref{sec:domain_analyticity},
state-dependent information can justify an
$\mathcal{O}(1)$ strip of analyticity for many
physically relevant states, even as the universal
zero-free wedge narrows with increasing rank.
Thus, increasing system size need not force this
threshold to decrease. For the conformal parameters
used throughout this work, all continuation benchmarks
in the main text use fewer R\'enyi inputs than the
estimated threshold $N^\ast\approx15$.

\subsection{Why simple poles?}
Finally, we also briefly comment on the chosen order of pole in the discrepancy function $D_{\alpha}(z)$. More generally, one may define
\be
D_{\alpha, \gamma}(z) = \frac{S_z - \alpha}{(z-1)^{\gamma}}, \ \ \gamma > 0
\ee
which diverges at the von Neumann point $z = 1$. In practice, however, a simple pole, $\gamma\simeq 1$, consistently yields the best performance. This is demonstrated in Fig.~\ref{fig:pole_order_CFT_1interval}, using the single-interval CFT example as a controlled benchmark.

\end{document}